\documentclass[12pt]{amsart}
\usepackage{amsmath,amssymb}
\usepackage{amsthm}
\usepackage{filecontents, tikz}
\usepackage{framed}
\usepackage[a4paper]{geometry}
\usepackage{enumerate}
\usepackage{url}

\title[Ising model on a regular matroid ]{A polynomial-time algorithm for estimating the partition
function of the ferromagnetic Ising model on a regular matroid
  }
\author{Leslie Ann Goldberg}
\author{Mark Jerrum}

\address{Department of Computer Science\\
University of Liverpool\\Ashton Building\\
Liverpool L69 3BX\\
United Kingdom.}

\address{School of Mathematical Sciences\\
Queen Mary, University of London\\
Mile End Road\\
London E1 NS\\
United Kingdom.}

\thanks{The work described in this paper was partly supported by EPSRC Research Grant
(refs\ EP/I011528/1 \& EP/I011935/1) ``Computational Counting''.
 An extended abstract appeared in the proceedings of ICALP 2011.
}

\def\easycompute{\zeta} 
\def\smallconst{\rho}
 
\def\newbigconst{66}

\def\calM{\mathcal{M}}

\def\Ztilde{\widetilde{Z}}

\def\R{\mathbb{R}}
\def\N{\mathbb{N}}

\def\Zow#1#2{\Ztilde(#1/p_1\backslash p_2,p_3;#2)}
\def\Ztw#1#2{\Ztilde(#1/p_2\backslash p_1,p_3;#2)}
\def\Zthw#1#2{\Ztilde(#1/p_3\backslash p_1,p_2;#2)}
\def\Zdw#1#2{\Ztilde(#1\backslash T;#2)}
\def\Zcw#1#2{\Ztilde(#1/T;#2)}
\def\Zjw#1#2{\Ztilde(#1/p_j\backslash (T-p_j);#2)}
\def\cor{c}
\def\zhat{\hat z}
\def\tr{\mathrm{T}}
\def\rest{|}
\def\pred{\varphi}
\def\Ex{\mathop\mathrm{\null E}\nolimits}
\def\deltasum{\bigtriangleup}

\def\rhs{right-hand side}
\def\lhs{left-hand side}
\def\bgamma{\boldsymbol\gamma}

\def\calC{\mathcal{C}}
\def\tTr{{T_\mathrm{decomp}}}
\def\tJS{{T_\mathrm{base}}}
\def\aTr{{\alpha_\mathrm{decomp}}}
\def\aJS{{\alpha_\mathrm{base}}}
\def\Itwo{\mathcal{I}_2}
\def\Ithree{\mathcal{I}_3}
\def\bdelta{\boldsymbol{\delta}}
\def\ZPotts{Z_\mathrm{Potts}}

\let\epsilon=\varepsilon
\let\rho=\varrho
\def\setminus{-}

\newtheorem{theorem}{Theorem}
\newtheorem{lemma}[theorem]{Lemma}
\newtheorem{corollary}[theorem]{Corollary}
\newtheorem{observation}[theorem]{Observation}
\newtheorem{remark}{Remark}
\newtheorem{definition}[theorem]{Definition}

\begin{document}

\maketitle

\begin{abstract}
We investigate the computational difficulty of approximating the partition function of the ferromagnetic Ising model on
a regular matroid. Jerrum and Sinclair have shown that there is a fully polynomial randomised approximation scheme (FPRAS)
for the class of graphic matroids. On the other hand, the authors have previously shown, subject to a complexity-theoretic
assumption, that there is no FPRAS for the class of binary matroids, which is a proper superset of the class of graphic matroids. 
In order to  map out the region where approximation is
feasible, we focus on the class of regular matroids, an important class of matroids
which properly includes the class of graphic matroids, and is properly included
in the class of binary matroids. Using Seymour's decomposition theorem, we give an FPRAS for the class of regular matroids.
\end{abstract}

%\begin{keywords} 
%approximation algorithms, Ising model, matroids, Tutte polynomial 
%\end{keywords}

%\begin{AMS}
%05B35,  %Matroids, geometric lattices
%05C31,  	%Graph polynomials
%68Q25,  %Analysis of algorithms and problem complexity
%68R10  	%Graph theory 
%68W25,  %Approximation algorithms
%68W40   %Analysis of algorithms 
%\end{AMS}

%\pagestyle{myheadings}
%\thispagestyle{plain}
%\markboth{L. A. GOLDBERG and M. JERRUM}{ISING MODEL ON A REGULAR MATROID}

\section{Introduction}

Classically, the Potts model~\cite{Potts} in statistical physics is defined on a graph.
Let $q$ be a positive integer and $G=(V,E)$ a graph with edge weights 
$\bgamma=\{\gamma_e:e\in E\}$;  the 
weight $\gamma_e>-1$ represents a ``strength of interaction'' along edge~$e$.  
The $q$-state Potts partition function specified by this weighted graph is
\begin{equation}\label{eq:PottsGph}
\ZPotts(G;q,\bgamma) = 
\sum_{\sigma:V\rightarrow [q]}
\prod_{e=\{u,v\}\in E}
\big(1+\gamma_e\,\delta(\sigma(u) ,\sigma(v))\big),
\end{equation}
where $[q]=\{1,\ldots,q\}$ is a set of $q$~spins or colours,
and $\delta(s,s')$ is~$1$ if $s=s'$, and 0 otherwise.
The partition function is a sum over ``configurations''~$\sigma$
which assign spins to vertices in all possible ways.
We are concerned with the computational complexity of approximately
evaluating the partition function~(\ref{eq:PottsGph}) and generalisations
of it.  For reasons that will become apparent shortly, we shall be concentrating on
the case $q=2$, which is the familiar Ising model.
In this special case, the two spins correspond to two possible magnetisations
at a vertex (or ``site''), and the edges (or ``bonds'') model interactions 
between sites.  In the {\it ferromagnetic\/} case, when $\gamma_e>0$, for all $e\in E$, 
the configurations $\sigma$ with many adjacent like spins
make a greater contribution to the partition function $\ZPotts(G;q,\bgamma)$ than
those with few;  in the {\it antiferromagnetic\/}
case, when  
$-1<\gamma_e<0$, 
the opposite is the case.

An equivalent way of looking at (\ref{eq:PottsGph}) is 
as a restriction of the (multivariate) 
Tutte polynomial, which is defined as follows:
\begin{equation}\label{eq:TutteGph}
\Ztilde(G;q,\bgamma)=\sum_{A\subseteq E}\gamma_A\,q^{\kappa(A)-|V|},
\end{equation}
where $\gamma_A=\prod_{e\in A}\gamma_e$ and $\kappa(A)$ 
denotes the number of connected components in
the graph $(V,A)$~\cite[(1.2)]{SokalMulti}.
This is perhaps not the most usual expression for the multivariate
Tutte polynomial of a graph, but it conveniently generalises to 
the Tutte polynomial of a matroid~(\ref{eq:TutteDef}), also \cite[(1.3)]{SokalMulti}, 
which is the main subject of this article.
Although (\ref{eq:TutteGph}) and (\ref{eq:PottsGph}) are
formally quite different, they agree when $q$ is a positive integer,
up to a factor of~$q^{-|V|}$. 

Jaeger, Vertigan and Welsh~\cite{JVW90} were the first to consider the
computational complexity of computing the Tutte polynomial.
They considered the classical bivariate Tutte polynomial in which 
the edge weights are constant, i.e., $\gamma_e=\gamma$ for all $e\in E$.
Their approach was to fix $q$ and~$\gamma$, and consider the computational
complexity of computing~(\ref{eq:TutteGph}) as a function of the instance
graph~$G$.  Jaeger et al.\ showed, amongst other things, that computing the
Tutte polynomial exactly is \#P-hard when $q>1$ and $\gamma\in(-1,\infty)-\{0\}$.
In particular, this means that 
the partition function~(\ref{eq:PottsGph}) of the Potts model
is computationally intractable, unless $\mathrm{P}=\mathrm{\#P}$.

In the light of this intractability result, it is natural to consider the 
complexity of approximate computation in the sense of ``fully polynomial
approximation scheme'' or FPRAS\null.
Before stating the known results, we quickly define the relevant concepts.
A \emph{randomised approximation scheme\/} is an algorithm for
approximately computing the value of a function~$f:\Sigma^*\rightarrow
\mathbb{R}$.
The
approximation scheme has a parameter~$\epsilon>0$ which specifies
the error tolerance.
A \emph{randomised approximation scheme\/} for~$f$ is a
randomised algorithm that takes as input an instance $ x\in
\Sigma^{* }$ (e.g.,  for the problem of
computing the bivariate Tutte polynomial of a graph, the
graph~$G$) and a rational error
tolerance $\epsilon >0$, and outputs a rational number $z$
(a random variable of the ``coin tosses'' made by the algorithm)
such that, for every instance~$x$,
\begin{equation}
\label{eq:3:FPRASerrorprob}
\Pr \big[e^{-\epsilon} f(x)\leq z \leq e^\epsilon f(x)\big]\geq \frac{3}{4}\, .
\end{equation}
The randomised approximation scheme is said to be a
\emph{fully polynomial randomised approximation scheme},
or \emph{FPRAS},
if it runs in time bounded by a polynomial
in $ |x| $ and $ \epsilon^{-1} $.
Note that the quantity $3/4$ in
Equation~(\ref{eq:3:FPRASerrorprob})
could be changed to any value in the open
interval $(\frac12,1)$ without changing the set of problems
that have randomised approximation schemes~\cite[Lemma~6.1]{JVV86}.

For the problem of computing the bivariate Tutte polynomial
in the antiferromagnetic situation, i.e., $-1<\gamma<0$, there is no
FPRAS for~(\ref{eq:TutteGph}) unless $\mathrm{RP}=\mathrm{NP}$~\cite{tuttepaper}.
This perhaps does not come as a great surprise, since, in the special case 
$\gamma=-1$, the equivalent expression~(\ref{eq:PottsGph}) counts proper colourings of a graph.
So we are led to consider the ferromagnetic case, $\gamma>0$.  Even here,
Goldberg and Jerrum~\cite{ferropotts} have recently provided evidence of computational intractability
(under a complexity theoretic assumption that is stronger than 
$\mathrm{RP}\not=\mathrm{NP}$) when $q>2$. 

The sequence of results so far described suggest we should focus on the 
special case $q=2$ and $\gamma>0$, i.e., the ferromagnetic Ising model.
Here, at last, there is a  positive result to report, as 
Jerrum and Sinclair~\cite{JSIsing} have presented an FPRAS for the partition
function~(\ref{eq:TutteGph}), with $q=2$ and arbitrary positive weights~$\bgamma$.
As hinted earlier, the Tutte polynomial makes perfect sense in the much wider 
context of an arbitrary matroid (see Sections \ref{sec:matroid} and~\ref{sec:decomp}
for a quick survey of matroid basics, and of the Tutte polynomial in the context
of matroids).  
The Tutte polynomial of a graph is merely the special case where the matroid is
restricted to be graphic.  It is natural to ask whether the positive result of~\cite{JSIsing}
extends to a wider class of matroids than graphic.  One extension of graphic
matroids is to the class of binary matroids.  Goldberg and Jerrum~\cite{GJbinary} recently 
provided evidence of computational intractability of the ferromagnetic Ising
model on binary matroids, under the same strong complexity-theoretic assumption
mentioned earlier.  
 
Sandwiched between the graphic (computationally easy) and binary matroids (apparently
computationally hard) is the class of regular matroids.  
Since it is interesting to
locate the exact boundary of tractability, we consider here the computational complexity
of estimating (in the FPRAS sense) the partition function of the Ising model 
on a regular matroid.  We show that there is an FRPAS in this situation. (See Section~\ref{sec:matroid} for 
matroid definitions and Section~\ref{sec:decomp} for the definition of
the Tutte polynomial $\Ztilde(\calM;q,\bgamma)$ of a matroid~$\calM$.)
\begin{theorem}\label{thm:main}
There is an FPRAS for the following problem. 
%\begin{description}
%\item[Instance] A binary matrix representing a regular matroid~$\calM$.
%A set $\bgamma=\{\gamma_e:e\in E(\calM)\}$ of non-negative rational edge weights.
%\item[Output] $\Ztilde(\calM;2,\bgamma)$.
An instance of the problem is a binary matrix representing a regular matroid~$\calM$,
and a set $\bgamma=\{\gamma_e:e\in E(\calM)\}$ of non-negative rational  weights
for elements of the ground set.
The required output is $\Ztilde(\calM;2,\bgamma)$.
%\end{description}
\end{theorem}

Aside from the existing FPRAS for graphic matroids, which also works,
by duality, for so-called cographic matroids, the main ingredient
in our algorithm is Seymour's decomposition theorem for 
regular matroids.  This theorem 
has been applied on at least one previous occasion 
to the design of a polynomial-time algorithm.
Golynski and Horton~\cite{GolynskiHorton} use the approach in their algorithm
for finding a minimum-weight basis of the cycle space (or circuit space)
of a regular matroid.
The decomposition theorem states that every regular matroid is either
graphic, cographic, a special matroid on 10 elements named $R_{10}$,
or can be decomposed as a certain kind of sum (called ``1-sum'', ``2-sum'' 
or ``3-sum'') of two smaller 
regular matroids (see Theorem~\ref{thm:Seymour}).  Since we know how to
handle the base cases (graphic, cographic and $R_{10}$), it seems likely
that the decomposition theorem will yield a polynomial time algorithm 
quite directly.  However, there is a catch (which does not arise in~\cite{GolynskiHorton}).  When we pull apart a regular
matroid into two smaller ones, say into the 3-sum of $\calM_1$ and~$\calM_2$, 
four subproblems are generated for each
of the parts $\calM_1$ and~$\calM_2$.  This is fine if the decomposition is fairly
balanced at each step, but that is not always the case.  In the case 
of a highly unbalanced decomposition, we face 
a
combinatorial explosion.

The solution we adopt is to ``solve'' recursively only the smaller 
subproblem, say~$\calM_2$.  Then we construct a constant size matroid~$\Ithree$
that 
we show
is equivalent to $\calM_2$ in the context of any 3-sum.  We then glue
$\Ithree$ onto~$\calM_1$, using the 3-sum operation, in place of~$\calM_2$.
%The matroid $\Ithree$ is small, just six elements, and has the property that
%forming the 3-sum with $\calM_1$ leaves $\calM_1$ unchanged as a matroid,
%though it acquires some new weights from~$\Ithree$.  
The matroid $\Ithree$ is both graphic and cographic and is small, having just six elements.
Most significantly, it has the property that
forming the 3-sum with $\calM_1$ leaves $\calM_1$ unchanged as a matroid,
though it acquires some new weights from~$\Ithree$.  
(In a sense, $\Ithree$ is an identity for the 3-sum operation.)
Then we just have to find the partition function of $\calM_1$ (with 
amended weights).  Since we have four recursive calls on ``small'' 
subproblems, but only one on  a ``large'' one, 
we not only avoid infinite recursion but we achieve polynomially
bounded running time.
%we achieve polynomially bounded
%running time.
The fact that $\Ithree$ is able to simulate the behaviour of an arbitrary 
regular matroid in the context of a 3-sum is a fortunate accident of the 
specialisation to $q=2$.  
Since the existing algorithm for the graphic case also
makes 
particular use of the fact that $q=2$, one gets the 
impression that the Ising model has very special properties compared 
with the Potts model in general.

\section{Matroid preliminaries}\label{sec:matroid}

A matroid is a combinatorial structure that has a number of equivalent 
definitions, but the one in terms of a rank function is the most
natural here.  A~set $E$ (the ``ground set'') 
together with a rank function 
$r:2^E\to\N$ 
is 
said to be a {\it matroid\/} if the following conditions are satisfied for
all subsets $A,B\subseteq E$: (i)~$0\leq r(A)\leq |A|$, (ii)~$A\subseteq B$
implies $r(A)\leq r(B)$ (monotonicity), and (iii)~$r(A\cup B)+r(A\cap B)
\leq r(A)+r(B)$ (submodularity).  A subset $A\subseteq E$ satisfying $r(A)=|A|$
is said to be {\it independent\/};  a maximal (with respect to inclusion)
independent set is a {\it basis}, and a minimal dependent set is a {\it circuit}.
A~circuit with one element is a {\it loop}.
We denote the ground set of matroid $\calM$ by $E(\calM)$ and its rank function
by~$r_\calM$.  
To every matroid~$\calM$ there is a {\it dual matroid\/}~$\calM^*$
with the same ground set $E=E(\calM)$ but rank function $r_{\calM^*}$ 
given by $r_{\calM^*}(A)=|A|+r_{\calM}(E-A)-r_{\calM}(E)$.
A {\it cocircuit\/} in $\calM$ is a set that is a circuit in~$\calM^*$;
equivalently, a cocircuit is a minimal set that intersects every basis.
A~cocircuit with one element is a {\it coloop}. A thorough exposition of the fundamentals (and beyond) of matroid theory can be found
in Oxley's book~\cite{OxleyBook}.

Important operations on matroids include contraction and deletion.  Suppose 
$T\subseteq E$ is any subset of the ground set of matroid~$\calM$.  
The {\it contraction
$\calM/T$ of $T$ from $\calM$} is the matroid on ground set $E\setminus T$ with 
rank function given by $r_{\calM/T}(A)=r_\calM(A\cup T)-r_\calM(T)$, for
all $A\subseteq E\setminus T$.
The {\it deletion
$\calM\backslash T$ of\/ $T$ from $\calM$} is the matroid on ground set $E\setminus T$ with 
rank function given by $r_{\calM \backslash T}(A)=r_\calM(A)$, for
all $A\subseteq E\setminus T$.  These operations are often combined, and we write
$\calM/T\backslash S$ for the matroid obtained by contracting $T$ from~$\calM$
and then deleting $S$ from the result.  The operations of contraction
and deletion are dual in the sense that $(\calM\backslash T)^*=\calM^*/T$.
For compactness, we shall often miss out set brackets, 
writing $\calM/p_1\backslash p_2,p_3$, for example, 
in place of $\calM/\{p_1\}\backslash\{p_2,p_3\}$.
The {\it restriction\/} $\calM\rest S$ of $\calM$ to  
$S\subseteq E$ is the matroid on ground set~$S$ that inherits
its rank function from~$\calM$; another way of expressing this 
is to say $\calM\rest S=\calM\backslash(E\setminus S)$. 

The matroid axioms are intended to abstract the notion of linear independence
of vectors.  Some matroids can be represented concretely as a matrix~$M$ with 
entries from a field~$K$, the columns of the matrix being identified with the
elements of~$E$.  The rank $r(A)$ of a subset $A\subseteq E$ is then just the 
rank of the submatrix of $M$ formed from columns picked out by~$A$.  A matroid
that can be specified in this way is said to be {\it representable over~$K$}.
A matroid that is representable over $\mathrm{GF}(2)$ is {\it binary}, and one that 
is representable over
every field
is {\it regular\/};  the regular matroids form
a proper subclass of binary matroids.  Another important class of matroids are
ones that arise as the {\it cycle matroid\/} of an undirected (multi)graph $G=(V,E)$.  
Here, the edge set~$E$ of the graph 
forms the ground set of the matroid, and the rank of a subset $A\subseteq E$ is
defined to be $r(A)=|V|-\kappa(A)$, where $\kappa(A)$ is the number of 
connected components in the subgraph $(V,A)$.  
A matroid is {\it graphic\/} if it arises as the cycle matroid of some graph,
and cographic if its dual is graphic.
Both the class of graphic matroids and the class of cographic matroids
are strictly contained in the class of regular matroids.

Now for some definitions more specific to the work in this article.
A {\it cycle\/} in a matroid is any subset of 
the
ground set that 
can be expressed as a disjoint union of circuits.  
We let $\calC(\calM)$ denote the set of cycles of a matroid~$\calM$.
If $\calM$ is
a binary matroid, the symmetric difference of any two cycles is again
a cycle~\cite[Thm.~9.1.2(vi)]{OxleyBook}, so 
$\calC(\calM)$,  viewed as a set of characteristic
vectors on $E(\calM)$, forms a vector
space over $\mathrm{GF}(2)$, which we refer to as the
{\it circuit space\/}  of $\calM$.  
Indeed, any vector space generated 
by a set of vectors in $\mathrm{GF}(2)^E$ can be regarded 
as the circuit space of a binary matroid on ground set~$E$
(see the remark following \cite[Cor.~9.2.3]{OxleyBook}).
The set of cycles $\calC(\calM)$ of a matroid~$\calM$ determines the set of circuits of~$\calM$
(these being just the minimal non-empty cycles),
which in turn determines~$\calM$.  
Note, from the definition of ``cycle''~\cite[\S1.3]{OxleyBook}, that
\begin{equation}\label{eq:cycleRestriction}
\calC(\calM{\rest}(E\setminus T))=\calC(\calM\backslash T)
 = \{C \in \calC(\calM) \mid C \subseteq E\setminus T\}.
\end{equation}
The term ``cycle'' in this context is not widespread, but is used by Seymour~\cite{SeymourDecomp}
in his work on matroid decomposition;  the term ``circuit space'' is more standard.

Consider two binary matroids $\calM_1$ and $\calM_2$ 
with $E(\calM_1)=E_1\cup T$ and $E(\calM_2) = E_2\cup T$ and $E_1\cap E_2=\emptyset$.
The {\it delta-sum\/} of $\calM_1$ and $\calM_2$ is the matroid 
$\calM_1\deltasum\calM_2$ on ground set $E_1\cup E_2$ 
with the following circuit space:
$$
\calC(\calM_1\deltasum\calM_2)=\big\{C\subseteq E_1\cup E_2:C=C_1\oplus C_2
\text{ for some $C_1\in\calC(\calM_1)$ and 
$C_2\in\calC(\calM_2)$}\big\},
$$
where $\oplus$ denotes symmetric difference~\cite{SeymourDecomp}.  
(The \rhs\ of the above equation defines a vector space over 
$\mathrm{GF}(2)$, and hence does describe the circuit space of some
matroid.) 
 
The special case of the delta-sum when $T=\emptyset$ is called the {\it 1-sum\/} of $\calM_1$ and~$\calM_2$.
(It is just the direct sum of the two matroids.)
The special case when $T$ is a singleton that is not a loop or coloop in $\calM_1$ or $\calM_2$,
and $|E(\calM_1)|,|E(\calM_2)|\geq3$, is called the {\it 2-sum\/} of $\calM_1$ or $\calM_2$
and is denoted $\calM_1\oplus_2\calM_2$.
Finally, the special case when $|T|=3$, $T$ is a circuit in both $\calM_1$ and $\calM_2$
but contains no co-circuit of either, 
and $|E(\calM_1)|,|E(\calM_2)|\geq7$, is called the {\it 3-sum\/} of $\calM_1$ or $\calM_2$
and is denoted $\calM_1\oplus_3\calM_2$~\cite{SeymourDecomp}.  

Our main tool is the following 
celebrated result  of Seymour~\cite[(14.3)]{SeymourDecomp}:
\begin{theorem}\label{thm:Seymour}
Every regular matroid $\calM$ may be constructed by means of 1- 2- and 3-sums,
starting with matroids each isomorphic to a minor of $\calM$, and each either 
graphic, cographic or isomorphic to a certain 10-element matroid $R_{10}$.
\end{theorem}
It is important for us that a polynomial-time algorithm exists for finding the 
decomposition promised by this theorem.  
As Truemper notes~\cite{Truemperpaper}, such an algorithm is given implicitly in Seymour's paper.
However, much more efficient algorithms are known.
In particular,   \cite[(10.6.1)]{TruemperBook} gives
a   polynomial-time algorithm for testing whether a binary matroid is graphic or cographic.
Also, \cite{Truemperpaper} 
gives a  cubic algorithm
for expressing any  
regular matroid that is not graphic, cographic or isomorphic to $R_{10}$
as a 1-sum, 2-sum or 3-sum.
For an exposition of this result, \cite[(8.4.1)]{TruemperBook}  
gives a polynomial-time algorithm that takes as input
a matrix representing a binary matroid and produces a 1-sum or a 2-sum decomposition, or, if 
 the 
matroid is 3-connnected, produces a 3-sum decomposition if one exists.
From the proof of \cite[(14.3)]{SeymourDecomp}, and the conditions on the matroid, (that it be regular, but
not graphic, cographic or isomorphic to $R_{10}$),
the matroid has a 3-separation in which the parts are sufficiently large, so
by \cite[(8.3.12)]{TruemperBook}, the desired 3-sum decomposition does exist (so is constructed).  
 
 Truemper's definitions are slightly different from Seymour's (which we use). However,
 it is easy to check that Truemper's 1-sum \cite[Section 8.2]{TruemperBook} is
 a 1-sum in the sense of Seymour. 
 It can also be checked that Truemper's 2-sum is a 2-sum in the sense of Seymour.
(For this it helps to note
that  the ground-set element in
 $E(\calM_1)\cap E(\calM_2)$ is represented by the column corresponding to element~``$x$''
 in 
 Truemper's matrix~$I \mid B^1$ and by the column corresponding to element~``$y$'' in 
 his matrix~$I \mid B^2$.)
Truemper's 3-sums \cite[Section 8.2]{TruemperBook} are not quite 3-sums in the sense of Seymour, but, as Golynski and Horton have noted~\cite{GolynskiHorton}, it is easy to 
apply an exchange operation \cite[Section 8.5]{TruemperBook} (in linear time) to
obtain a 3-sum in the sense of Seymour. 
  
While Seymour's decomposition theorem is in terms of 1-sums, 2-sums and 3-sums, it will be convenient for us to 
 do our preparatory work, in the following section, using the slightly more general notion of a delta-sum.

\section{Tutte polynomial and decomposition}
\label{sec:decomp}
Suppose $\calM$ is a matroid, $q$ is an indeterminate, and $\bgamma=\{\gamma_e:e\in E(\calM)\}$ 
is a collection of indeterminates, indexed by elements of the ground set of~$\calM$.
The (multivariate) {\it Tutte polynomial\/}
of $\calM$ and~$\bgamma$ is defined to be
\begin{equation}\label{eq:TutteDef}
\Ztilde(\calM;q,\bgamma)=\sum_{A\subseteq E(\calM)}\gamma_A\,q^{-r_\calM(A)},
\end{equation}
where $\gamma_A=\prod_{e\in A}\gamma_e$~\cite[(1.3)]{SokalMulti}.
In this article, we are interested in the {\it Ising model}, which 
corresponds to the specialisation of the above polynomial to $q=2$,
so we shall usually omit the parameter~$q$ in the above notation
and assume that $q$ is set to~2.
As we are concerned with approximate computation, we shall invariably
be working in a environment 
in which
each of the 
indeterminates~$\gamma_e$ is assigned some real value;  
furthermore, we shall be focusing on the {\it ferromagnetic\/} case, in which
those values or {\it weights\/} are all non-negative.
For convenience, we refer to the pair $(\calM,\bgamma)$ as a ``weighted matroid'',
and we will assume throughout that $\gamma_e\geq0$ for all $e\in E(\calM)$.
For background material on the multivariate Tutte polynomial, and its relation 
to the classical 2-variable Tutte polynomial, refer to Sokal's 
expository article~\cite{SokalMulti}. 

In order to exploit Theorem~\ref{thm:Seymour}, 
we need to investigate how definition~(\ref{eq:TutteDef})
behaves 
when $\calM = \calM_1 \deltasum \calM_2$.
Suppose that $E_1$, $T $ and $E_2$ are mutually disjoint sets 
and consider two binary matroids $\calM_1$ and $\calM_2$ 
with $E(\calM_1)=E_1\cup T$ and $E(\calM_2) = E_2\cup T$.
\begin{definition}
For $A\subseteq E_i$ and $S\subseteq T$,  let
$e_i(A,S) = r_{\calM_i}(A\cup S) - r_{\calM_i}(A)$.
For   $A\subseteq E_1$ and $B\subseteq E_2$,
let 
$\cor(A,B) =  r_{\calM}(A\cup B) - r_{\calM_1}(A) - r_{\calM_2}(B) $.
\end{definition}
Intuitively,  $e_i(A,S)$ is the ``excess'' rank that~$S$ adds to~$A$ in $\calM_i$
and
$\cor(A,B)$ is the appropriate ``correction'' to the rank function of~$\calM$ relative
to the rank function of the  
1-sum of~$\calM_1$ and~$\calM_2$,
which assigns rank
$r_{\calM_1}(A) + r_{\calM_2}(B)$ to
$A\cup B$.

Suppose 
$A\subseteq E_1$ and $B\subseteq E_2$.
Consider the vector spaces $V_1,V_2$ and $V$ over $\mathrm{GF}(2)$ 
on $A\cup B\cup T$ defined as follows:
\begin{align*}V_1&=\calC(\calM_1\rest(A\cup T)),\\
V_2&=\calC(\calM_2\rest(B\cup T)),\text{ and} \\
V&=V_1+V_2,
\end{align*}
which is the span of $V_1 \cup V_2$.
Let $\Pi_T:V\to V$ denote the projection of $A\cup B\cup T$ onto $T$.  That is, if $x\in V$ is 
the characteristic vector of some set $X\subseteq A\cup B\cup T$, then 
$\Pi_T(x)$ is the characteristic vector of~$X\cap T$.  Informally, $\Pi_T$
sets to zero everything outside~$T$.  
From now on, we won't distinguish subsets of $A\cup B\cup T$ and their characteristic vectors.
Let $\Pi_T(V)$ denote the range of~$\Pi_T$.

\begin{lemma}\label{lem:rankOfAUB}
Consider two binary matroids $\calM_1$ and $\calM_2$ with
$E(\calM_1) = E_1 \cup T$ and $E(\calM_2)= E_2 \cup T$. Let
$\calM = \calM_1 \Delta \calM_2$. Suppose $A\subseteq E_1$
and $B \subseteq E_2$.
With $V_1$, $V_2$, $V$ and $\Pi_T$ as above, 
\begin{equation}\label{eq:cident}
c(A,B)= e_1(A,T)+e_2(B,T)+\dim(V_1\cap V_2)+\dim\Pi_T(V) -2|T|.
\end{equation}
Specialising to the 2-sum, 
\begin{equation}\label{eq:key2}
c(A,B)=e_1(A,T)+e_2(B,T)+\dim\Pi_T(V)-2,
\end{equation}
and specialising to the 3-sum,
\begin{equation}\label{eq:key3}
c(A,B)=e_1(A,T)+e_2(B,T)+\dim\Pi_T(V)-5.
\end{equation}
\end{lemma}

\begin{proof}
From the definition of delta-sum, and noting (\ref{eq:cycleRestriction}),
the kernel $\ker\Pi_T$ of~$\Pi_T$ (i.e., all vectors in $V$ that map to the zero vector
under~$\Pi_T$) is the cycle space of the matroid $\calM\rest (A\cup B)$.  
Thus,
$$
\dim\calC(\calM\rest (A\cup B))=\dim\ker\Pi_T=\dim V-\dim\Pi_T(V),
$$
where the second equality is from the Rank-nullity Theorem~\cite[\S50 Thm~1]{Halmos}.

\def\im{\mathop\mathrm{im}\nolimits}
We need the following relationship between the rank of a matroid $\calM'$
on ground set~$E$ 
and the dimension of its cycle space:
\begin{equation}\label{eq:orthog}
\dim\calC(\calM')+r_{\calM'}(E)=|E|.
\end{equation}
This identity comes from the observation that 
$\calC(\calM')$ 
contains precisely those 
vectors that are orthogonal to all rows of the matrix $M$ 
representing~$\calM'$.
So 
$\calC(\calM')$ 
is the orthogonal complement of the row space of~$M$, 
and the dimensions of these two spaces sum to $|E|$, 
the number of columns in~$M$~\cite[\S17 Thm~1]{Halmos}.
Applying~(\ref{eq:orthog}) in turn to the matroids $\calM\rest (A\cup B)$, $\calM_1\rest(A\cup T)$
and $\calM_2\rest(B\cup T)$, we obtain
\begin{align}
r_\calM(A\cup B)&=|A\cup B| -\dim\calC(\calM\rest (A\cup B))=|A|+|B|-\dim V+\dim\Pi_T(V),\label{rM}\\
r_{\calM_1}(A\cup T)&=|A\cup T|-\dim\calC(\calM_1\rest (A\cup T))=|A|+|T|-\dim V_1,\text{ and}\notag\\
r_{\calM_2}(B\cup T)&=|B\cup T|-\dim\calC(\calM_2\rest (B\cup T))=|B|+|T|-\dim V_2.\notag
\end{align}
It follows immediately from the last two equations that 
\begin{align}
r_{\calM_1}(A)&=|A|+|T|-\dim V_1-e_1(A,T),\text{ and}\label{rM1}\\
r_{\calM_2}(B)&=|B|+|T|-\dim V_2-e_2(B,T).\label{rM2}
\end{align}
Subtracting (\ref{rM1}) and (\ref{rM2}) from (\ref{rM}),
\begin{align*}
c(A,B)&=r_\calM(A\cup B)-r_{\calM_1}(A)-r_{\calM_2}(B)\\
&= e_1(A,T)+e_2(B,T)+\dim V_1+\dim V_2 -\dim V +\dim\Pi_T(V) -2|T|\\
&= e_1(A,T)+e_2(B,T)+\dim(V_1\cap V_2)+\dim\Pi_T(V) -2|T|,
\end{align*}
where in the final equality we use the fact~\cite[\S12 Ex.~7]{Halmos} that
$$\dim(V_1+V_2)=\dim V_1+\dim V_2 -\dim(V_1\cap V_2).$$  
This gives (\ref{eq:cident}).

Specialising (\ref{eq:cident}) to the 2-sum, $|T|=1$ and $\dim(V_1\cap V_2)=0$ 
(since the only cycle common to $\calM_1$ and $\calM_2$ is the empty cycle),
giving~(\ref{eq:key2}).
Specialising to the 3-sum, $|T|=3$ and $\dim(V_1\cap V_2)=1$ 
(since the cycles common to $\calM_1$ and $\calM_2$ are the empty cycle 
and $T$ itself), giving~(\ref{eq:key3}).
\end{proof}

The point about (\ref{eq:key2}) and (\ref{eq:key3}) is that they are easy to compute with,
since $\dim\Pi_T(V)$ has a very clean combinatorial interpretation.  
Observe that $\Pi_T(V)$ can be viewed as the set 
\begin{equation}\label{eq:PiT}
\big\{S\subseteq T: S=(C_1\oplus C_2)\cap T\text{ where } C_1\in\calC(\calM_1\rest (A\cup T)), 
C_2\in\calC(\calM_2\rest (B\cup T))\big\}.
\end{equation}
We will see, when we come to apply Lemma~\ref{lem:rankOfAUB}, that the dimension
of this set, viewed as a vector space, is easy to calculate.

The interaction between the classical {\it bivariate\/} Tutte polynomial
and 2- and 3-sums has been investigated previously, for example by
Andrzejak~\cite{Andrzejak}.  In principle, it would be possible to 
assure oneself that his proof carries over from the bivariate to the
multivariate situation, and then recover the appropriate formulas
(Lemmas \ref{lem:2sumsplit} and~\ref{lem:3sumsplit} below) by appropriate
algebraic translations.  However, in the case of the 3-sum there is
an obstacle.  While Andrzejak's identities remain valid, as identities of
rational functions on~$q$, they become degenerate (through division
by zero) under the specialisation $q=2$.  In fact, this degeneracy is 
crucial to us, in that it reduces the dimension of the bilinear form 
in Lemma~\ref{lem:3sumsplit} from five, as in Andrzejak's general result,
to four, in our special formula for $q=2$.
In light of these considerations, and because certain intermediate results 
in our proofs are in any case required later, we derive the required 
formulas here.  First 
we give the
the formula for (a slight relaxation of) the 2-sum.
 
\begin{lemma}\label{lem:2sumsplit}
Suppose  $(\calM_1,\bgamma_1)$ and $(\calM_2,\bgamma_2)$ are weighted binary 
matroids with $E(\calM_1)\cap E(\calM_2)=\{p\}$,
where $p$ is not a loop in either $\calM_1$ or~$\calM_2$.  
Denote by $\bgamma=\bgamma_1\deltasum\bgamma_2$ 
the weighting on $E(\calM_1\deltasum\calM_2)$
inherited from $\calM_1$ and $\calM_2$.
Then
\begin{equation}\label{eq:2sumsplit}
\Ztilde(\calM_1\deltasum\calM_2;\bgamma)=
(\Ztilde(\calM_1\backslash p;\bgamma_1),\Ztilde(\calM_1/p;\bgamma_1))
\begin{pmatrix}2&-1\\-1&1\end{pmatrix}
\begin{pmatrix}\Ztilde(\calM_2\backslash p;\bgamma_2)\\ \Ztilde(\calM_2/p;\bgamma_2)\end{pmatrix}.
\end{equation}
\end{lemma}

\begin{proof}
The machinery of  the proof is a little heavy in relation to the 
scale of the result, but its use will provide a warm-up for the proof of the 
analogous result for 3-sums (Lemma~\ref{lem:3sumsplit}).

Let $E_1 = E(\calM_1)\setminus \{p\}$ and $E_2 = E(\calM_2) \setminus\{p\}$.
Use  the definition of $c(A,B)$
to write the Tutte polynomial of the delta-sum as 
\begin{equation}\label{eq:2sumExpr}
\Ztilde(\calM_1\deltasum \calM_2;\bgamma) = 
\sum_{A\subseteq E_1} \sum_{B\subseteq E_2}\gamma_A \gamma_B
q^{-r_{\calM_1}(A) -r_{\calM_2}(B)}
q^{-\cor(A,B)}.
\end{equation} 
We are in the situation of Lemma~\ref{lem:rankOfAUB}, specifically (\ref{eq:key2}),
with $T=\{p\}$.  Recall that $\Pi_T(V)$ is given by~(\ref{eq:PiT}).
We claim that there exists a cycle $C\in \calC(\calM_1\rest (A\cup\{p\}))$ 
(respectively $C\in \calC(\calM_2\rest (B\cup\{p\}))$)
with $p\in C$
if and only if $e_1(A,\{p\})=0$ (respectively $e_2(B, \{p\})=0$).  
For the ``if'' direction, suppose 
that $e_1(A,\{p\})=0$.  Let $A'\subseteq A$ be any maximum independent set 
in $\calM_1\rest A$.  Since $e_1(A,\{p\})=0$, the set $A\cup\{p\}$ is dependent,
so $A'\cup\{p\}$ is dependent.
Take any circuit~$C$ (minimal dependent set) in 
$A'\cup \{p\}$ 
and note that~$C$
must contain~$p$.  For the ``only if'' direction, suppose $e_1(A,\{p\})=1$.
Observe that $p$ cannot be part 
of any circuit~$C\in \calC(\calM_1\rest (A\cup\{p\}))$, as removing $p$ from~$C$
will always reduce the rank by one, contradicting minimality. 
Pulling all these observations together, we see that $\Pi_T(V)=\{\emptyset\}$
if $e_1(A,\{p\})=1$ and $e_2(B,\{p\})=1$, and $\Pi_T(V)=\{\emptyset,\{p\}\}$
otherwise.  In the former case $\dim\Pi_T(V)=0$, and in the latter $\dim\Pi_T(V)=1$.
Substituting into (\ref{eq:key2}),
$$
\cor(A,B)=\begin{cases}
   -1,&\text{iff $e_1(A,\{p\})=0$ and $e_2(B,\{p\})=0$;}\\
   0,&\text{otherwise}.
\end{cases}
$$

Thus, letting $\zhat$ be the column vector
$\zhat^i=(\zhat^i_0,\zhat^i_1)^\tr$, for $i\in\{1,2\}$, where
$$
\zhat^i_0=\sum_{A:A\subseteq E_i,e_i(A,\{p\})=0}\gamma_A q^{-r_{\calM_i}(A)}
\quad\text{and}\quad
\zhat^i_1=\sum_{A:A\subseteq E_i,e_i(A,\{p\})=1}\gamma_A q^{-r_{\calM_i}(A)},
$$
we may express identity (\ref{eq:2sumExpr}) as 
\begin{equation}\label{eq:lhs2sumsplit}
\Ztilde(\calM_1\deltasum \calM_2;\bgamma) = (\zhat^1)^\tr C \zhat^2,
\end{equation}
where
$$
C =\begin{pmatrix}q&1\\1&1\end{pmatrix} = \begin{pmatrix}2&1\\1&1\end{pmatrix}.
$$

Having expressed the \lhs\ of (\ref{eq:2sumsplit}) in terms of $\zhat^1$ and $\zhat^2$,
we now wish to do the same for the \rhs.
For $i\in\{1,2\}$, let 
$$z^i=(z^i_0,z^i_1)^\tr=(\Ztilde(\calM_i\backslash p;\bgamma_i),\Ztilde(\calM_i/p;\bgamma_i))^\tr,$$
and observe that 
$$  \Ztilde(\calM_i\backslash p;\bgamma_i)
 = \sum_{A \subseteq E_i} \gamma_A q^{-r_{\calM_i}(A)}$$
and
$$  \Ztilde(\calM_i/p;\bgamma_i)
= \sum_{A \subseteq E_i} \gamma_A q^{-r_{\calM_i}(A)}
q^{1-e_i(A,\{p\})},$$ so

\begin{equation}\label{eq:coordchange'}
z^i=V\zhat^i,
\end{equation}
where 
$$
V=\begin{pmatrix}1&1\\q&1\end{pmatrix} = \begin{pmatrix}1&1\\2&1\end{pmatrix}.
$$
(We are using the fact that $\bgamma$ agrees with $\bgamma_1$ on $E_1$
and with $\bgamma_2$ on~$E_2$.)
Now, $\zhat^i_0\geq0$ and $\zhat^i_1>0$ (since 
the weights $\gamma_e$ are non-negative and
the term $A=\emptyset$ contributes 
positively to the sum defining~$\zhat^i_1$).  Together with (\ref{eq:coordchange'}) these
facts imply 
\begin{equation}\label{eq:z0z1ineqs}
z^i_0\leq z^i_1< 2z^i_0,
\end{equation}
an inequality we shall need later.  With
$$
D=\begin{pmatrix}2&-1\\-1&1\end{pmatrix},
$$
we can express the \rhs\ of (\ref{eq:2sumsplit}), using transformation (\ref{eq:coordchange'}), as 
\begin{equation}\label{eq:rhs2sumsplit}
(z^1)^\tr D z^2=(V \zhat^1)^\tr DV\zhat^2 =(\zhat^1)^\tr (V^\tr DV)\zhat^2.
\end{equation}
Comparing (\ref{eq:lhs2sumsplit}) and (\ref{eq:rhs2sumsplit}), we see that
to complete the proof, we just need to verify the matrix equation
$C=V^\tr DV$
with $C$ and $V$ specialised to $q=2$, and this is easily
done.  
\end{proof}

The same calculation can be carried through for 
(a slight relaxation of) the 3-sum.  Denote by $D$ the matrix
$$
D = \begin{pmatrix}
 4 & -1 & -1 & -1\\
 -1 & 1 & 0 & 0\\
 -1 & 0 & 1 & 0 \\
 -1 & 0 & 0 & 1\\
\end{pmatrix}.
$$

\begin{lemma}\label{lem:3sumsplit}
Suppose that $(\calM_1,\bgamma_1)$ and $(\calM_2,\bgamma_2)$ are weighted
binary matroids with 
$E(\calM_1)\cap E(\calM_2)=T=\{p_1,p_2,p_3\}$, and suppose 
also that $T$ is a circuit in both $\calM_1$ and $\calM_2$.  For $i=1,2$, let 
$$z^i=\big(\Zdw{\calM_i}{\bgamma_i},\Zow{\calM_i}{\bgamma_i},
\Ztw{\calM_i}{\bgamma_i},\Zthw{\calM_i}{\bgamma_i}\big)^\tr.$$
Denote by $\bgamma=\bgamma_1\deltasum\bgamma_2$ the weighting on $E(\calM_1\deltasum\calM_2)$
inherited from $\calM_1$ and $\calM_2$.
Then 
\begin{equation}\label{eq:3sumsplit}
\Ztilde(\calM_1\deltasum\calM_2;\bgamma)=(z^1)^\tr Dz^2.
\end{equation}
\end{lemma}

\begin{proof}
Use the definition of $c(A,B)$ 
to write the Tutte polynomial of   the delta-sum as 
\begin{equation}
\Ztilde(\calM_1\deltasum\calM_2;\bgamma) = 
\sum_{A\subseteq E_1} \sum_{B\subseteq E_2}\gamma_A \gamma_B
q^{-r_{\calM_1}(A) -r_{\calM_2}(B)}
q^{-\cor(A,B)}.
\end{equation} 
Now, to construct an equation analogous to~(\ref{eq:2sumsplit}), we will
work with the minors
$\calM_i / S_1 \backslash S_2$
for partitions~$(S_1,S_2)$ of~$T$.
This gives eight different minors of~$\calM_i$, but some of them are equivalent.
In particular,  since $T$ is a circuit of~$\calM_i$
(which means that $p_1$, $p_2$ and $p_3$ are linearly dependent in the binary matrix
representing $\calM_i$
but any pair of them is independent),
any partition~$(S_1,S_2)$ with $|S_1| \geq  2$
gives a matroid $\calM_i / S_1 \backslash S_2$
which is equivalent to 
$\calM_i /  T $.
To see this, note that for $X\subseteq E_i$,
$r_{\calM_i / S_1 \backslash S_2}(X) = 
r_{\calM_i}(X\cup S_1) - r_{\calM_i}(S_1) =
r_{\calM_i}(X \cup  T) - 2
= 
r_{\calM_i}(X \cup  T) - r_{\calM_i}(T)$.
Thus, we will use the five minors with $|S_1|\neq 2$.
  
Let's collect the formulas for the Tutte polynomials of the minors that we'll need.
\begin{align}
\label{tutteminors}
   \Zdw{\calM_i}{\bgamma_i} &= \sum_{A \subseteq E_i} \gamma_A q^{-r_{\calM_i}(A)}\\
\nonumber
   \Zjw{\calM_i}{\bgamma_i} &= \sum_{A \subseteq E_i} \gamma_A q^{-r_{\calM_i}(A)}
q^{1-e_i(A,\{p_j\})},\quad\text{for $j\in\{1,2,3\}$,}\\
\noalign{\noindent and}
\nonumber
 \Zcw{\calM_i}{\bgamma_i} &= \sum_{A \subseteq E_i} \gamma_A 
 q^{-r_{\calM_i}(A)}q^{2-e_i(A,T)}
\end{align} 

With an eye to the proof of Lemma~\ref{lem:2sumsplit},
we next need to understand the function $e_i(A,S)$.  It turns out that,
with $A$ fixed, $e_i(A,S)$ is completely determined by its value  
on the singleton sets $S=\{p_1\},\{p_2\},\{p_3\}$.  
Also, there are only five possibilities
for the values of $e_i(A,S)$ on those singletons.
For $i=1,2$, the following collection of predicates $\{\pred^i_0,\ldots,\pred^i_4\}$
on $A\subseteq E_i$, captures those five possibilities.
\begin{enumerate}[$\mathrm{P}$1.]
\setcounter{enumi}{-1}
\item $\pred^i_0(A)$ iff $e_i(A,\{p_1\})=e_i(A,\{p_2\})=e_i(A,\{p_3\})=0$ (in which case, 
  $e_i(A,S)=0$ for all $S$).\label{none}
\item $\pred^i_1(A)$ iff $e_i(A,\{p_1\})=0$ and $e_i(A,\{p_2\})=e_i(A,\{p_3\})=1$ (in which case, 
  $e_i(A,S)=1$ when $|S|\geq2$).\label{one}
\item $\pred^i_2(A)$ iff $e_i(A,\{p_2\})=0$ and $e_i(A,\{p_1\})=e_i(A,\{p_3\})=1$ (in which case, 
  $e_i(A,S)=1$ when $|S|\geq2$).\label{two}
\item $\pred^i_3(A)$ iff $e_i(A,\{p_3\})=0$ and $e_i(A,\{p_1\})=e_i(A,\{p_2\})=1$ (in which case, 
  $e_i(A,S)=1$ when $|S|\geq2$).\label{three}
\item $\pred^i_4(A)$ iff $e_i(A,\{p_1\})=e_i(A,\{p_2\})=e_i(A,\{p_3\})=1$ (in which case, 
  $e_i(A,S)=2$ when $|S|\geq2$).\label{last}
\end{enumerate}
Of course, in all cases, $e_i(A,\emptyset)=0$.

First observe that it is not 
possible that exactly one of $e_i(A,\{p_1\})$, $e_i(A,\{p_2\})$ and $e_i(A,\{p_3\})$
to take the value~1, since this would imply that
two members of~$T$ are dependent on~$A$, and hence the third would be.
Thus, P\ref{none}--P\ref{last} are exhaustive as well as mutually exclusive.

We still need to verify the additional information provided in parentheses.
In case P\ref{none}, $e_i(A,S)=0$ for all $S$, since
all of the elements in~$T$ are dependent on~$A$.
In case P\ref{one} (P\ref{two} and P\ref{three} are symmetrical) $p_1$ is
dependent on~$A$, and so $e_i(A,\{p_1,p_2\})=
e_i(A,\{p_2\})=1$;  then $e_i(A,T)=e_i(A,\{p_1,p_2\})=1$, since $p_3$ depends
on $p_1$ and~$p_2$.
Similarly, $e_i(A,T)=e_i(A,\{p_2,p_3\})$, since $p_1$ depends
on $p_2$ and~$p_3$, so $e_i(A,\{p_2,p_3\})=1$.
The final parenthetical claim (in P\ref{last}) is slightly trickier, and relies
on the fact that the matroids we are working with are binary.  (There
is a simple counterexample based on the uniform matroid $U_{4,2}$ for the
claim in general.)  So consider the elements of $A$ 
as columns of the representing matrix (vectors
over $\mathrm{GF}(2)$), 
and similarly consider the elements of $S$ as columns.
We must rule out the possibility that $e_i(A, S) =1$
for some 2-element subset $S$, say $S=\{p_1,p_2\}$.
So suppose $e_i(A,\{p_1\}) = e_i(A,\{p_2\}) = e_i(A, \{p_1,p_2\}) = 1$.
Then $p_2$ is dependent on $A \cup \{p_1\}$, i.e.,
so (also viewing $p_1$ and $p_2$ as vectors over $\mathrm{GF}(2)$), 
$p_2 = p_1 + \sum_{e\in B} e$, for some $B\subseteq A$.
(Note that $p_1$ must be included in this expression, otherwise $p_2$
would be dependent on~$A$.)
Since $T$ is a circuit we know that $p_3 = p_1 + p_2$.
Now write $p_3 = p_1 + p_2 = p_1 + p_1 + \sum_{e \in B} e=\sum_{e \in B} e$. 
Thus $p_3$ is dependent on $A$, and $e_i(A,p_3)=0$, contrary to what we assumed.

By analogy with the proof of Lemma~\ref{lem:2sumsplit} 
define $\zhat^i=(\zhat^i_0,\ldots,\zhat^i_4)^\tr$ for $i=1,2$ by
$$
\zhat^i_k=\sum_{A:A\subseteq E_i,\pred^i_k(A)}\gamma_Aq^{-r_{\calM_i}(A)}
$$
and 
\begin{align*}
z^i&=\big(\Zdw{\calM_i}{\bgamma_i},\Zow{\calM_i}{\bgamma_i},\\
&\qquad\Ztw{\calM_i}{\bgamma_i},\Zthw{\calM_i}{\bgamma_i},\Zcw{\calM_i}{\bgamma_i}\big)^\tr.
\end{align*}
Note that the dimension here is one greater 
(five instead of four) than in the statement of the lemma,
but this extra dimension will disappear towards the end of the proof.
As before, we can re-express the \lhs\ of (\ref{eq:3sumsplit}) in terms of $\zhat^1,\zhat^2$.

Our first step will be to show that the value $c(A,B)$ depends on the
predicates $\pred^1_j(A)$ and $\pred^2_k(B)$ but does not depend otherwise on~$A$ and~$B$.
In other words, there exists a well-defined
$5\times5$ matrix $C$ with $C_{j,k}=q^{-c(A,B)}$, for $0\leq j,k\leq 4$,
where $A\subseteq E_1$ and $B\subseteq E_2$ are 
any sets satisfying $\pred^1_j(A)$ and $\pred^2_k(B)$.
Specifically, using Lemma~\ref{lem:rankOfAUB}, we will verify that
$$C =
\begin{pmatrix}
q^2 & q & q& q& 1\\
q & q & 1 & 1 & 1\\
q & 1 & q & 1 & 1\\
q & 1 & 1 & q & 1\\
1 & 1 & 1 & 1 & 1 
\end{pmatrix}
= \begin{pmatrix}
4 & 2& 2& 2& 1\\
2 & 2 & 1 & 1 & 1\\
2 & 1 & 2 & 1 & 1\\
2 & 1 & 1 & 2 & 1\\
1 & 1 & 1 & 1 & 1 
\end{pmatrix}.
$$

Consider the first row of the matrix~$C$.
Here, $\pred^1_0(A)$ holds, so that
$e_1(A,\{p_1\})=e_1(A,\{p_2\})=e_1(A,\{p_3\})=0$, and $e_1(A,T)=0$. 
Arguing as in Lemma~\ref{lem:2sumsplit}, 
there exist cycles $C_1,C_1',C_1''$ in $\calC(\calM_1\rest (A\cup T))$ 
satisfying $C_1\cap T=\{p_1\}$, $C_1'\cap T=\{p_2\}$ and $C_1''\cap T=\{p_3\}$. 
Since $\Pi_T(V)$ is a vector space it contains all linear combinations 
(via addition in $\mathrm{GF}(2)$) of these three, and hence $\dim\Pi_T(V)=3$.  
Substituting into (\ref{eq:key3}) we obtain
$$
c(A,B)=e_1(A,T)+e_2(B,T)+\dim\Pi_T(V)-5=e_2(B,T)-2,
$$
from which the first row (and, by symmetry, column) of $C$ can be read off.  

At the other extreme, consider the fifth and final row.  
Here $\pred^1_4(A)$ holds, so that
$e_1(A,\{p_1\})=e_1(A,\{p_2\})=e_1(A,\{p_3\})=1$, and $e_1(A,T)=2$. 
It follows that, for any $C_1\in\calC(\calM_1\rest (A\cup T))$, either
$C_1\cap T=\emptyset$ or $C_1\cap T=T$.  (One cannot have $C_1\cap T=\{p_1,p_2\}$,
for example, as that would imply the existence of a cycle $C'_1=C_1\oplus T$ with 
$C_1'\cap T=\{p_3\}$.)
In the second column, $\pred^2_1(B)$ holds, so 
$e_2(B,\{p_1\})=0$, 
$e_2(B,\{p_2\})=
e_2(B,\{p_3\})=1$ 
and 
$e_2(B,T)=1$. 
There are four possibilities for a cycle
$C_2\in\calC(\calM_2\rest (B\cup T))$, namely 
$C_2\cap T$ is $\emptyset$, $\{p_1\}$, $\{p_2,p_3\}$ or~$T$.  Thus $\dim\Pi_T(V)=2$,
and 
$$c(A,B)=e_1(A,T)+e_2(B,T)+\dim\Pi_T(V)-5=2+1+2-5=0.$$
A similar calculation is valid for the third and fourth columns.  In the final column, 
$\pred^2_4(B)$ 
holds, and $C_2\cap T$ is either $\emptyset$ or~$T$.
Thus $\dim\Pi_T(V)=1$ and 
$$c(A,B)=e_1(A,T)+e_2(B,T)+\dim\Pi_T(V)-5=2+2+1-5=0.$$
This deals with the final row (and, by symmetry, column) of~$C$.

It only remains to check the centre $3\times3$ submatrix of~$C$.  
There are two cases, diagonal and off-diagonal, and they can be 
verified using very similar arguments to the ones we used for the
boundary elements of the matrix.
In conclusion, the \lhs\ of (\ref{eq:3sumsplit}) may be expressed as
\begin{equation}\label{eq:lhs3sumsplit}
\Ztilde(\calM_1\deltasum\calM_2;\bgamma)=(\zhat^1)^\mathrm{T}C\zhat^2.
\end{equation}

Also as before, there is a linear relationship between $z^i$ and $\zhat^i$, expressed by
\begin{equation}
z^i=V\zhat^i,\label{eq:coordChange}
\end{equation}
where 
$$V =
\begin{pmatrix}
1 & 1 & 1 & 1 & 1\\
q & q & 1 & 1 & 1\\
q & 1 & q & 1 & 1\\
q & 1 & 1 & q & 1\\
q^2 & q & q & q & 1
\end{pmatrix}=
\begin{pmatrix}
1 & 1 & 1 & 1 & 1\\
2 & 2 & 1 & 1 & 1\\
2 & 1 & 2 & 1 & 1\\
2 & 1 & 1 & 2 & 1\\
4 & 2 & 2 & 2 & 1
\end{pmatrix}.
$$
It is straightforward to verify correctness of $V$ row by row.
For example, the third row expresses the easily checked fact that
$$\Ztw{\calM_i}{\bgamma_i}=q\zhat^i_0+\zhat^i_1+q\zhat^i_2+\zhat^i_3+\zhat^i_4.$$

Identity (\ref{eq:coordChange})
allows us to rewrite the \rhs\ of (\ref{eq:3sumsplit})
in terms of $\zhat^1$ and~$\zhat^2$.  Let $D'$ be the $5\times5$ matrix
obtained from $D$ in the statement of the theorem by padding it out to the
right and below by a single column and row of zeros.  Then the \rhs\
of (\ref{eq:3sumsplit}) may be written as
\begin{equation}\label{eq:rhs3sumsplit}
(z^1)^\tr D z^2=(V \zhat^1)^\tr D'V\zhat^2 =(\zhat^1)^\tr (V^\tr D'V)\zhat^2.
\end{equation}
Comparing (\ref{eq:lhs3sumsplit}) and (\ref{eq:rhs3sumsplit}), we see that
to complete the proof, we just need to verify the matrix equation
\begin{equation}\label{eq:matrixeq}
C=V^\tr D'V,
\end{equation}
with $C$ and $V$ specialised to $q=2$, and this is easily
done.  
\end{proof}

\begin{remark}
With $q$ regarded as an indeterminate, Equation~(\ref{eq:matrixeq})
has a unique solution in which $D'$ is a matrix of full rank, whose entries
are rational functions of~$q$.  When $q$ is specialised to~2, 
the rank of matrix~$C$ drops to~4, allowing a solution in which 
$D'$ also has rank~4.  When $q=2$, 
there is flexibility in the choice of the matrix~$D'$, which we 
exploit in order to reduce the dimension by~1.  
This is the reason that the \rhs\ of (\ref{eq:3sumsplit}) is of dimension~4,
whereas one would expect dimension~5 a priori.  This apparent accident
is crucial to the design of the algorithm.
\end{remark}

 \section{Signatures}
The goal of this section is to show that, for every weighted binary
weighted matroid $(\calM,\bgamma)$ with distinguished element~$p$, there is a small
(2-element) weighted matroid $(\Itwo,\bdelta)$ that is equivalent to $\calM$ in the
following sense:  if we replace $(\calM,\bgamma)$ 
by $(\Itwo,\bdelta)$ in the context of any 2-sum,
the Tutte polynomial (specialised to $q=2$) of the 2-sum is changed by a factor that
is independent of the context.  
Moreover, the weighted matroid $(\Itwo,\bdelta)$ can readily be computed given 
a ``signature'' of $\calM$.
There is also a 6-element weighted matroid $(\Ithree,\bdelta)$
that does a similar job for 3-sums.
 
To make this precise, let $\Itwo$ be the matroid with a 2-element ground set 
$\{p,e\}$ that form a 2-circuit.  
$\Itwo$ can be viewed as the cycle matroid of a graph consisting of two parallel edges, $p$ and $e$.
 
\begin{lemma}\label{lem:simsig2} Suppose $(\calM,\bgamma)$ is a weighted binary matroid with a distinguished element~$p$
that is not a loop.
Then there is a   weight $d\geq0$ such that, for 
every weight function $\bdelta$ with $\delta_e = d$,
$$\frac{\Ztilde(\Itwo/p;\bdelta)}{\Ztilde(\Itwo\backslash p;\bdelta)}=
\frac{\Ztilde(\calM/p;\bgamma)}{\Ztilde(\calM\backslash p;\bgamma)}.$$ 
The value~$d$ 
can be computed from
$\Ztilde(\calM\backslash p;\bgamma)$ and
$\Ztilde(\calM/p;\bgamma)$ --- it does not otherwise depend upon~$\calM$ or~$\bgamma$.
\end{lemma}

\begin{proof} 
Let $z_0=\Ztilde(\calM\backslash p;\bgamma)$ and $z_1=\Ztilde(\calM/p;\bgamma)$, and set $d=2(z_1-z_0)/(2z_0-z_1)$.
Inequality~(\ref{eq:z0z1ineqs}) implies that the numerator is non-negative, 
and the denominator is strictly positive.
Then
$$\Ztilde(\Itwo/p;\bdelta)/\Ztilde(\Itwo\backslash p;\bdelta)=(1+d)/(1+\tfrac12d)=z_1/z_0=
\Ztilde(\calM/p;\bgamma)/\Ztilde(\calM\backslash p;\bgamma).$$
\end{proof}

\begin{corollary}
\label{cor:2sumexact}
Suppose  $(\calM_1,\bgamma_1)$ and $(\calM_2,\bgamma_2)$ are weighted binary 
matroids with $E(\calM_1)\cap E(\calM_2)=\{p\}$,
where $p$ is not a loop in either $\calM_1$ or~$\calM_2$. 
Then there is a   weight $d\geq0$ such that, for 
every weight function $\bdelta$ with $\delta_e = d$,
$$\Ztilde(\calM_1\deltasum\calM_2;\bgamma_1\deltasum\bgamma_2)=\easycompute 
\Ztilde(\calM_1\deltasum\Itwo;\bgamma_1\deltasum\bdelta),$$
where $\easycompute = 2(2+d)^{-1}\Ztilde(\calM_2\backslash p;\bgamma_2)$.
The value~$d$ can be computed from
$\Ztilde(\calM_2\backslash p;\bgamma_2)$ and
$\Ztilde(\calM_2/p;\bgamma_2)$ --- it does not otherwise depend upon~$(\calM_1,\bgamma_1)$ or $(\calM_2,\bgamma_2)$.
\end{corollary}
 
\begin{proof} 
Let $s = \Ztilde(\calM_2/p;\bgamma_2) / \Ztilde(\calM_2\backslash p;\bgamma_2)$.
By Lemma~\ref{lem:2sumsplit},
$$\Ztilde(\calM_1\deltasum\calM_2;\bgamma_1 \deltasum \bgamma_2)=
(\Ztilde(\calM_1\backslash p;\bgamma_1),\Ztilde(\calM_1/p;\bgamma_1))
\begin{pmatrix}2&-1\\-1&1\end{pmatrix}
\begin{pmatrix} 1\\  s\end{pmatrix}\Ztilde(\calM_2\backslash p;\bgamma_2).$$
Similarly, using Lemma~\ref{lem:simsig2},
$$\Ztilde(\calM_1\deltasum \Itwo;\bgamma_1 \deltasum \bdelta)=
(\Ztilde(\calM_1\backslash p;\bgamma_1),\Ztilde(\calM_1/p;\bgamma_1))
\begin{pmatrix}2&-1\\-1&1\end{pmatrix}
\begin{pmatrix} 1\\  s\end{pmatrix}\Ztilde(\Itwo\backslash p;\bdelta).$$
Thus, $\easycompute = \Ztilde(\calM_2\backslash p;\bgamma_2) / \Ztilde(\Itwo\backslash p;\bdelta)$.
\end{proof}
 
It is crucial in Corollary~\ref{cor:2sumexact} that $\easycompute$ does not depend
at all on~$(\calM_1,\bgamma_1)$.
As we shall see presently, an analogous result holds for 3-sums, though the
calculations are more involved.
For a weighted binary matroid $(\calM,\bgamma)$ with distinguished elements $T=\{p_1,p_2,p_3\}$,
the {\it signature\/} of $(\calM,\bgamma)$ with respect to~$T$ is the vector
$$
\sigma(\calM;T,\bgamma)=\frac{
(\Zow{\calM}{\bgamma},
   \Ztw{\calM}{\bgamma},\Zthw{\calM}{\bgamma})   
   }{\Zdw{\calM}{\bgamma}}.
$$
What we are seeking in the 3-sum case is a small matroid whose signature
is equal to that of a given binary matroid~$\calM$.  Such a matroid will
be equivalent to $\calM$ in the context of any 3-sum.   
Before constructing such a matroid we   give two technical lemmas
that investigate inequalities between the Tutte polynomial of various minors 
of a binary matroid. These inequalities will restrict the domain of possible signatures that can occur.

\begin{lemma}\label{lem:ineqs1}
Suppose that $(\calM,\bgamma)$ is a weighted binary matroid 
with distinguished elements $T=\{p_1,p_2,p_3\}$,
and that $T$ is a circuit in $\calM$.  Let 
\begin{align*}
z^\tr&=(z_0,z_1,z_2,z_3,z_4)\\
&=\big(\Zdw{\calM}{\bgamma},\Zow{\calM}{\bgamma},\\
&\qquad\Ztw{\calM}{\bgamma},\Zthw{\calM}{\bgamma},\Zcw{\calM}{\bgamma}\big).
\end{align*}
Then the following (in)equalities hold:
(i)~$z_0>0$,
(ii)~$z_0\leq z_1,z_2,z_3\leq 2z_0$, (iii)~$\tfrac12z_4<z_1,z_2,z_3\leq z_4$
and (iv)~$z_1+z_2+z_3=2z_0+z_4$.
\end{lemma}

\begin{proof}
Inequality (i) follows from the definition (\ref{eq:TutteDef}) since the contribution of $A=\emptyset$ is~$1$ and the contribution
of every other~$A$ is non-negative.
Consider identity (\ref{eq:coordChange}).  Since $\zhat\geq0$ (coordinatewise),
the vector~$z$ in the statement of the lemma is in the cone generated by the 
columns of~$V$ (i.e., $z$ is a non-negative linear combination of columns).  
Denote the rows of $V$ by $V_0,\ldots,V_4$.  Then, interpreting all vector
inequalities coordinatewise, (ii) is a consequence
of $V_0\leq V_1,V_2,V_3\leq 2V_0$  (iii), with non-strict inequality, of 
$\frac12 V_4\leq V_1,V_2,V_3\leq V_4$, and (iv) of $V_1+V_2+V_3=2V_0+V_4$.
We know that $\zhat_4>0$, since the weights are non-negative and its expansion contains at least one non-zero 
term, namely the one corresponding to $A=\emptyset$;  this implies that the
first inequality in~(iii) must be strict.
\end{proof}

\begin{lemma}\label{lem:ineqs2}
Under the same assumptions as Lemma~\ref{lem:ineqs1}, $z_1z_2,z_1z_3,z_2z_3\leq z_0z_4$.
\end{lemma}

\begin{proof}
Let the ground set of $\calM$ be $E\cup T$ where
where $E\cap T = \emptyset$.
Recall the earlier definition
$$
\zhat_k=\sum_{A:A\subseteq E,\pred_k(A)}\gamma_A q^{-r_\calM(A)}, \quad\text{for $0\leq k\leq4$},
$$ 
and recall also that the the predicates $\pred_k$ are exhaustive and
mutually exclusive, so that 
$\Zdw{\calM}{\bgamma} = \zhat_0+\zhat_1+ \cdots +\zhat_4$.  
For brevity, write $\zhat=\Zdw{\calM}{\bgamma}$.

The rank function $r_\calM$ of a matroid is submodular, 
so $-r_\calM$ is supermodular
and the probability distribution $\mu$ 
(defined by $\mu(A)=q^{-r_\calM(A)}
\gamma_A
/\Ztilde(\calM \backslash T, \bgamma)$)
associated with the random cluster
model satisfies the condition for Fortuin-Kasteleyn-Ginibre (FKG)
inequality~\cite[Section~19, Theorem 5]{BolBook}, provided $q \geq 1$.  
Let $f = 1-e(A,\{p_1\})$ and $g = 1-e(A,\{p_2\})$. 
Note that these are both monotonically increasing, as $A$ grows.
(This can either be seen directly, or as a consequence of submodularity of~$r_\calM$.)
So the quantities $f$ and $g$ are positively correlated, i.e.,
$$\Pr\nolimits_\mu(fg=1)=\Ex_\mu(fg)\geq \Ex_\mu(f)\Ex_\mu(g)
   =\Pr\nolimits_\mu(f=1)\Pr\nolimits_\mu(g=1),$$
which is equivalent to 
$$
\frac{\zhat_0}{\zhat} \geq \frac{\zhat_0 + \zhat_1}{\zhat} \times \frac{\zhat_0 + \zhat_2}{\zhat},
$$
which in turn is equivalent to
\begin{equation}\label{eq:FKGconsequence}
   \zhat_0\zhat \geq   (\zhat_0 + \zhat_1)(\zhat_0 + \zhat_2).
\end{equation}

Consider the first inequality we are required to establish, namely 
$z_0z_4\geq z_1z_2$.  
Using the linear transformation (\ref{eq:coordChange}), we
may express this as an equivalent inequality in terms of~$\zhat_k$:
$$
 \zhat(4\zhat_0 + 2(\zhat_1 + \zhat_2 + \zhat_3) + \zhat_4)  
    \geq  (2\zhat_0 + 2\zhat_1 + \zhat_2 + \zhat_3 + \zhat_4)
    (2\zhat_0 + \zhat_1 + 2 \zhat_2 + \zhat_3 + \zhat_4),
$$
where the four linear factors may be read off from the appropriate
rows (first, last, second and third, respectively) of the matrix~$V$.
Applying the definition of~$\zhat$ we obtain the equivalent inequality 
$$
\zhat(\zhat + 3\zhat_0+\zhat_1 + \zhat_2 + \zhat_3)  \geq  
   (\zhat + \zhat_0 + \zhat_1)(\zhat + \zhat_0 + \zhat_2),
$$
which further simplifies, through cancellation, to
\begin{equation}\label{eq:veq''}
 \zhat_0\zhat + \zhat_3 \zhat  \geq  (\zhat_0 + \zhat_1)(\zhat_0 + \zhat_2).    
\end{equation}
Now (\ref{eq:FKGconsequence}) implies (\ref{eq:veq''}), and we are done,
since the other two advertised inequalities follow by symmetry.
\end{proof}

Denote by $\Ithree$ the 6-element matroid
with ground set $\{p_1,p_2,p_3,e_1,e_2,e_3\}$, 
whose circuit space  
is generated by the circuits $\{p_1,e_1\}$, $\{p_2,e_2\}$, $\{p_3,e_3\}$, 
and $\{p_1,p_2,p_3\}$.   The matroid $\Ithree$ can be thought of as the cycle matroid of
a certain graph, namely, the graph with  parallel pairs of edges 
$\{p_1,e_1\}$, $\{p_2,e_2\}$ and $\{p_3,e_3\}$
in which edges $p_1$, $p_2$ and $p_3$ form a length-3 cycle in the graph.

Let $T=\{p_1,p_2,p_3\}$.
We start by showing that, as long as a signature $(s_1,s_2,s_3)$ satisfies
certain equations, which Lemmas~\ref{lem:ineqs1} and~\ref{lem:ineqs2} will guarantee,
then it is straightforward to compute a weighting~$\bdelta$ so that the weighted matroid $(\Ithree,\bdelta)$ has
signature $\sigma(\Ithree;T,\bgamma)=(s_1,s_2,s_3)$.

\begin{lemma}\label{lem:simsig3}
Suppose $s_1$, $s_2$ and $s_3$ satisfy
\begin{equation}\label{eq:sig1}
2+s_1-s_2-s_3>0,\quad
2-s_1+s_2-s_3>0,\quad
2-s_1-s_2+s_3>0 ,\end{equation} 
\begin{equation}\label{eq:sig2}
s_1+s_2+s_3-3\geq0, 
\end{equation}
\begin{equation}\label{eq:sig3}
s_1+s_2+s_3-s_2s_3-2\geq0,\quad
s_1+s_2+s_3-s_1s_3-2\geq0, \mbox{ and}\quad
s_1+s_2+s_3-s_1s_2-2\geq0;
\end{equation}
then there are non-negative weights $d_1$, $d_2$ and $d_3$
such that, for any weight 
function $\bdelta$ with $\delta_{e_1}=d_1$, 
$\delta_{e_2} = d_2$ and $\delta_{e_3} = d_3$,
$\sigma(\Ithree;T,\bdelta)=(s_1,s_2,s_3)$.
The values~$d_1$, $d_2$ and~$d_3$ can be computed  from $s_1$, $s_2$ and $s_3$.
\end{lemma}

\begin{proof} 
  
Define 
\begin{align*}
S_1&=2+s_1-s_2-s_3\\
S_2&=2-s_1+s_2-s_3\\
S_3&=2-s_1-s_2+s_3\\
R&=s_1+s_2+s_3-2.
\end{align*}

Define the weights $d_1,d_2,d_3$ for~$\bdelta$ as follows:
\begin{align*}
d_1&=-1+\sqrt{RS_1/S_2S_3}\\
d_2&=-1+\sqrt{RS_2/S_1S_3}\\
d_3&=-1+\sqrt{RS_3/S_1S_2},
\end{align*}
 
By inequalities (\ref{eq:sig1}) and (\ref{eq:sig2}), $R$, $S_1$, $S_2$ and $S_3$  are all strictly positive.
Thus 
$d_1,d_2,d_3$ are well defined.
Finally
$$
RS_1-S_2S_3=4(s_1+s_2+s_3-s_2s_3-2)\geq0,
$$
where the inequality is~(\ref{eq:sig3}).  Thus $d_1$, and hence,
by symmetry, $d_2$ and $d_3$, are all non-negative. 

Let $Y = d_1 d_2 + d_1 d_3 + d_2 d_3 + d_1 d_2 d_3$,
and note that 
\begin{align*}
\Zdw{\Ithree}{\bdelta}&= q^{-0} + q^{-1}(d_1 + d_2 + d_3) + q^{-2} Y, \mbox{ and}\\
\Zow{\Ithree}{\bdelta}&= q^{-0}(1+d_1) + q^{-1}(d_2+d_3 + Y).
\end{align*}
Substituting for $d_1,d_2,d_3$ in these expressions, 
using the helpful identity
$$Y+d_1+d_2+d_3+1= (1+d_1)(1+d_2)(1+d_3)=R\sqrt{R/S_1S_2S_3},$$
we obtain
\begin{align*}
\Zdw{\Ithree}{\bdelta}&=\tfrac14(R+S_1+S_2+S_3)\sqrt{R/S_1S_2S_3}=\sqrt{R/S_1S_2S_3},\\
\noalign{\noindent\text{and}}
\Zow{\Ithree}{\bdelta}&=\tfrac12(R+S_1)\sqrt{R/S_1S_2S_3}=s_1\sqrt{R/S_1S_2S_3},
\end{align*}
and hence $\Zow{\Ithree}{\bdelta}/\Zdw{\Ithree}{\bdelta}=s_1$.  By symmetry, similar
identities hold for the other components of the signature of~$(\Ithree,\bdelta)$.
Summarising, $\sigma(\Ithree;T,\bdelta)=(s_1,s_2,s_3)$,
as desired.
\end{proof}

Temporarily leaving aside the issue of approximation, the way that we will use Lemma~\ref{lem:simsig3} is
captured in the following corollary, which is analogous to Corollary~\ref{cor:2sumexact}.

\begin{corollary}\label{cor:3sum}
Suppose that $(\calM_1,\bgamma_1)$ and $(\calM_2,\bgamma_2)$ are weighted
binary matroids with 
$E(\calM_1)\cap E(\calM_2)=T=\{p_1,p_2,p_3\}$, and suppose 
also that $T$ is a circuit in both $\calM_1$ and $\calM_2$.  
Let $(s_1,s_2,s_3) = \sigma(\calM_2;T,\bgamma_2)$.
Then there are non-negative weights $d_1$, $d_2$ and $d_3$
such that, for any weight 
function $\bdelta$ with $\delta_{e_1}=d_1$, 
$\delta_{e_2} = d_2$ and $\delta_{e_3} = d_3$,
$$
\Ztilde(\calM_1\deltasum\calM_2;\bgamma_1 \deltasum\bgamma_2)=\easycompute \Ztilde(\calM_1 \deltasum \Ithree;\bgamma_1 \deltasum \bdelta),
$$ 
where $$\easycompute=\Zdw{\calM_2}{\bgamma_2}/\Zdw{\Ithree}{\bdelta}=
\Zdw{\calM_2}{\bgamma_2} / \sqrt{R/S_1S_2S_3},$$
in the notation of the proof of Lemma~\ref{lem:simsig3}.
The values~$d_1$, $d_2$ and~$d_3$ can be computed  from $s_1$, $s_2$ and $s_3$ --- they do
not otherwise depend upon $(\calM_1,\bgamma_1)$ or $(\calM_2,\bgamma_2)$.
Moreover, the values~$R$, $S_1$, $S_2$ and $S_3$ are byproducts of this computation.
\end{corollary}

\begin{proof}
As before, for $i=1,2$, let
$$z^i=\big(\Zdw{\calM_i}{\bgamma_i},\Zow{\calM_i}{\bgamma_i},
\Ztw{\calM_i}{\bgamma_i},\Zthw{\calM_i}{\bgamma_i}\big)^\tr.$$
Let $s = (1,s_1,s_2,s_3)^\tr$.
By Lemma~\ref{lem:3sumsplit},
$$\Ztilde(\calM_1\deltasum\calM_2;\bgamma_1\deltasum \bgamma_2)=(z^1)^\tr Dz^2 =
\Ztilde(\calM_2\backslash T;\bgamma_2) (z^1)^\tr D s.$$
Similarly, if $\sigma(\Ithree;T,\bdelta)=(s_1,s_2,s_3)$, then
$$\Ztilde(\calM_1\deltasum\Ithree;\bgamma_1\deltasum \bdelta)= 
\Ztilde(\Ithree\backslash T;\bdelta) (z^1)^\tr D s.$$

To simplify the notation (and avoid confusing the index ``$2$'' in $z^2$ with an exponent), let $z$ denote the vector $z^2$.
Now, from inequalities (iii) and (iv) of Lemma~\ref{lem:ineqs1},  we see that 
$2(z_0+z_1)> z_1+z_2+z_3$, so Equation~(\ref{eq:sig1}) is satisfied.
By inequality~(ii) of Lemma~\ref{lem:ineqs1}, Equation~(\ref{eq:sig2}) is satisfied.
Finally, from Lemma~\ref{lem:ineqs2} and identity (iv) of Lemma~\ref{lem:ineqs1},
$z_2z_3\leq z_0(z_1+z_2+z_3-2z_0)$, so Equation~(\ref{eq:sig3}) is satisfied. So, by Lemma~\ref{lem:simsig3},
the weights $d_1$, $d_2$ and $d_3$ can be computed and 
$\sigma(\Ithree;T,\bdelta)=(s_1,s_2,s_3)$.
The value of $\Ztilde(\Ithree \backslash T;\bdelta)$ is calculated in the proof of Lemma~\ref{lem:simsig3}.
\end{proof}
 
The problem with using Corollary~\ref{cor:3sum}
to replace the complicated expression $ \Ztilde(\calM_1\deltasum\calM_2;\bgamma_1 \deltasum\bgamma_2)$
with the simpler $\Ztilde(\calM_1 \deltasum \Ithree;\bgamma_1 \deltasum \bdelta)$
is that, in general, we will not be able to compute the necessary values~$s_1$, $s_2$ and~$s_3$.
Instead, we will use our FPRAS recursively to approximate these values.
Thus, we need a version of Corollary~\ref{cor:3sum} that accommodates some approximation error.
Unfortunately, this creates some technical complexities.
We will use the following lemma.

\begin{lemma}\label{lem:approx3sum}
Suppose that $(\calM_1,\bgamma_1)$ and $(\calM_2,\bgamma_2)$ are weighted
binary matroids with 
$E(\calM_1)\cap E(\calM_2)=T=\{p_1,p_2,p_3\}$, and suppose 
also that $T$ is a circuit in both $\calM_1$ and $\calM_2$.  

Suppose that $\epsilon \leq 1$
and that $\smallconst$ is a sufficiently small positive constant ($\smallconst = 1/6000$ will do).
Suppose that, coordinatewise, 
\begin{equation}
\label{eq:appxclose}
e^{-\epsilon \smallconst}\sigma(\calM_2;T,\bgamma_2) \leq (\tilde s_1,\tilde s_2,\tilde s_3) 
\leq e^{\epsilon \smallconst}\sigma(\calM_2;T,\bgamma_2).
\end{equation} 
Then there are non-negative weights $d_1$, $d_2$ and $d_3$
such that, for any weight 
function $\bdelta$ with $\delta_{e_1}=d_1$, 
$\delta_{e_2} = d_2$ and $\delta_{e_3} = d_3$,
$$
e^{-\epsilon}\Ztilde(\calM_1\deltasum\calM_2;\bgamma_1 \deltasum\bgamma_2)
\leq \easycompute \Ztilde(\calM_1 \deltasum \Ithree;\bgamma_1 \deltasum \bdelta)
\leq e^{\epsilon}\Ztilde(\calM_1\deltasum\calM_2;\bgamma_1 \deltasum\bgamma_2),$$
 
where $$\easycompute=
\Zdw{\calM_2}{\bgamma_2} / \sqrt{R/S_1S_2S_3},$$
in the notation of the proof of Lemma~\ref{lem:simsig3}.
The values~$d_1$, $d_2$ and~$d_3$ can be computed  from $\tilde s_1$, $\tilde s_2$ and $\tilde s_3$ --- they do
not otherwise depend upon $(\calM_1,\bgamma_1)$ or $(\calM_2,\bgamma_2)$.
Moreover, the values~$R$, $S_1$, $S_2$ and $S_3$ are byproducts of this computation.
\end{lemma}

\begin{proof}
Let   
$(r_1,r_2,r_3) = \sigma(\calM_2;T,\bgamma_2)$.
First, using Lemma~\ref{lem:XXX} in the appendix with $\chi=\epsilon \smallconst$, 
we can use $\tilde s_1$, $\tilde s_2$ and $\tilde s_3$ to
compute $s_1$, $s_2$ and $s_3$ satisfying
\begin{equation}
\label{eq:near} e^{-\newbigconst\epsilon \smallconst} s_i \leq  r_i \leq e^{\newbigconst\epsilon\smallconst} s_i,
\end{equation} 
and Equations~(\ref{eq:sig1}), (\ref{eq:sig2}) and (\ref{eq:sig3}).
To see that Lemma~\ref{lem:XXX} applies, note that
Equation~(\ref{E3}) follows from Lemma~\ref{lem:ineqs1} (ii).
Equation~(\ref{E1}) is analogous to Equation~(\ref{eq:sig1}) and is established in the same way as Equation~(\ref{eq:sig1}) is
established in the proof of Corollary~\ref{cor:3sum}. Similarly, Equation~(\ref{E2}) is analogous to Equation~(\ref{eq:sig3}) and is
established as in the proof of Corollary~\ref{cor:3sum}. (It is not clear that $r_1 \leq r_2 \leq r_3$, but this can be thought of as a renaming
inside Lemma~\ref{lem:XXX}. Furthermore, the corresponding assumption $\tilde s_1 \leq \tilde s_2 \leq \tilde s_3$ is without loss of generality,
since, if the assumption does not hold, then these values can be swapped without violating the proximity of $\tilde s_i$ to $r_i$.)

By Lemma~\ref{lem:simsig3}, $s_1$, $s_2$ and $s_3$ can be used to compute
non-negative weights $d_1$, $d_2$ and $d_3$
such that, for any weight 
function $\bdelta$ with $\delta_{e_1}=d_1$, 
$\delta_{e_2} = d_2$ and $\delta_{e_3} = d_3$,
$\sigma(\Ithree;T,\bdelta)=(s_1,s_2,s_3)$. As in the proof of Corollary~\ref{cor:3sum}, this implies
$$\Ztilde(\calM_1\deltasum\Ithree;\bgamma_1\deltasum \bdelta)= 
\Ztilde(\Ithree\backslash T;\bdelta) (z^1)^\tr D s,$$
where $s = (1,s_1,s_2,s_3)^\tr$ and $\Ztilde(\Ithree\backslash T;\bdelta) = \sqrt{R/S_1S_2S_3}$ for some byproducts $R$, 
$S_1$, $S_2$ and $S_3$ of the computation.

Now let $r=(1,r_1,r_2,r_3)^\tr$. Then, as in the proof of Corollary~\ref{cor:3sum}, 
$$\Ztilde(\calM_1\deltasum\calM_2;\bgamma_1\deltasum \bgamma_2)= 
\Ztilde(\calM_2\backslash T;\bgamma_2) (z^1)^\tr D r.$$
Also, since $\newbigconst \smallconst \leq 1/56$,
$e^{-\epsilon /56} s_i \leq  r_i \leq e^{\epsilon/56} s_i$.
The result then follows from Lemma~\ref{lem:YYY} in the appendix, since  
$z$, $s$ and $r$ have positive entries, and satisfy
$1 \leq z_i / z_0 \leq 2$,
$1 \leq s_i / s_0 \leq 2$, and $1 \leq r_i/r_0 \leq 2$ for $i\in\{1,2,3\}$ by Lemma~\ref{lem:ineqs1} (i) and (ii).
\end{proof}

We also need a lemma, similar to Lemma~\ref{lem:approx3sum}, that
is appropriate for 2-sum.

  \begin{lemma}
\label{lem:approx2sum}
Suppose  $(\calM_1,\bgamma_1)$ and $(\calM_2,\bgamma_2)$ are weighted binary 
matroids with $E(\calM_1)\cap E(\calM_2)=\{p\}$,
where $p$ is not a loop in either $\calM_1$ or~$\calM_2$.
Let $z_0 = \Ztilde(\calM_2\backslash p;\bgamma_2)$ and let $z_1 = \Ztilde(\calM_2/p;\bgamma_2)$.
Suppose that $\epsilon \leq 1$ and that $\smallconst$ is a sufficiently small positive constant (as in Lemma~\ref{lem:approx3sum} --- here it
suffices to take $\smallconst \leq 1/160$).
Suppose that $e^{-\epsilon \smallconst} z_i \leq \tilde z_i \leq e^{\epsilon \smallconst} z_i$ for $i\in\{0,1\}$.
Then there is a   weight $d\geq0$ such that, for 
every weight function $\bdelta$ with $\delta_e = d$,
$$e^{-\epsilon}
\Ztilde(\calM_1\deltasum\calM_2;\bgamma_1\deltasum\bgamma_2)
\leq\easycompute \Ztilde(\calM_1\deltasum\Itwo;\bgamma_1\deltasum\bdelta)
\leq e^{\epsilon}\Ztilde(\calM_1\deltasum\calM_2;\bgamma_1\deltasum\bgamma_2),$$
where $\easycompute = 2(2+d)^{-1} z_0$.
The value~$d$ can be computed from $\tilde z_0$ and $\tilde z_1$. It does not otherwise depend upon~$(\calM_1,\bgamma_1)$ or $(\calM_2,\bgamma_2)$.
\end{lemma}

\begin{proof}

This is similar to the proof of Lemma~\ref{lem:approx3sum}, but easier.
First, observe that $z_0>0$ and $1\leq z_1/z_0 \leq 2$  (Equation (\ref{eq:z0z1ineqs})).
The first step is to use $\tilde z_0$ and $\tilde z_1$ to compute $z'_0$ and $z'_1$
such that $z'_0>0$ and $1\leq z'_1/z'_0 \leq 2$ and $e^{-4 \epsilon \smallconst} z'_i \leq z_i \leq e^{4 \epsilon \smallconst} z'_i$.
This is straightforward. Just set $z'_0 = \tilde z_0 e^{\epsilon \smallconst}$.
If $\tilde z_1 \leq z_0'$ then set $z'_1 = z'_0$.
If $\tilde z_1 \geq 2 z_0'$ then set $z'_1 = 2 z'_0$.
Otherwise, set $z'_1 = \tilde z_1$. (Note that 
$\tfrac{\tilde z_1}{\tilde z_0} \geq e^{-2\epsilon\smallconst} \tfrac{z_1}{z_0} \geq e^{-2\epsilon\smallconst}$, so
if $\tilde z_1 \leq z_0'$ then 
$z'_1 =   e^{\epsilon \smallconst} \tilde z_0 \leq e^{\epsilon \smallconst} e^{2 \epsilon \smallconst} \tilde z_1 \leq
e^{4 \epsilon \smallconst} z_1$. 
Similarly, $\tfrac{\tilde z_1}{\tilde z_0} \leq e^{2\epsilon \smallconst} \tfrac{z_1}{z_0} \leq 2e^{2\epsilon \smallconst}$,
so if $\tilde z_1 \geq 2 z_0'$ then $z'_1 = 2e^{\epsilon \smallconst} \tilde z_0
\geq e^{\epsilon \smallconst}e^{-2 \epsilon \smallconst} \tilde z_1 \geq e^{-2\epsilon\smallconst} z_1$.)
As in the proof of Lemma~\ref{lem:simsig2},  
let $d=2(z'_1-z'_0)/(2z'_0-z'_1)$. Let $s' = z'_1/z'_0$ and note
that, for 
every weight function~$\bdelta$ with $\delta_e = d$,
$\frac{\Ztilde(\Itwo/p;\bdelta)}{\Ztilde(\Itwo\backslash p;\bdelta)}= s'$. Let $s = z_1/z_0$.
As in the proof of Corollary~\ref{cor:2sumexact}, note that
$$\Ztilde(\calM_1\deltasum\calM_2;\bgamma_1 \deltasum \bgamma_2)=
\big(\Ztilde(\calM_1\backslash p;\bgamma_1),\Ztilde(\calM_1/p;\bgamma_1)\big)  \begin{pmatrix}2&-1\\-1&1\end{pmatrix}
\begin{pmatrix} 1\\  s\end{pmatrix} z_0,$$
and
$$\Ztilde(\calM_1\deltasum \Itwo;\bgamma_1 \deltasum \bdelta)=
\big(\Ztilde(\calM_1\backslash p;\bgamma_1),\Ztilde(\calM_1/p;\bgamma_1)\big) \begin{pmatrix}2&-1\\-1&1\end{pmatrix}
\begin{pmatrix} 1\\  s'\end{pmatrix} (1+\tfrac{d}{2}).$$
Since $e^{-8\epsilon \smallconst} s' \leq s \leq e^{8\epsilon \smallconst} s'$, similar to the proof of Lemma~\ref{lem:YYY},
\begin{align*}
e^{-160 \epsilon \smallconst}
\big(\Ztilde(\calM_1\backslash p;\bgamma_1), &\Ztilde(\calM_1/p;\bgamma_1)\big)  \begin{pmatrix}2&-1\\-1&1\end{pmatrix}
\begin{pmatrix} 1\\  s\end{pmatrix}\\
&\leq 
\big(\Ztilde(\calM_1\backslash p;\bgamma_1),\Ztilde(\calM_1/p;\bgamma_1)\big) \begin{pmatrix}2&-1\\-1&1\end{pmatrix}
\begin{pmatrix} 1\\  s'\end{pmatrix}\\
&\leq e^{160 \epsilon \smallconst}
\big(\Ztilde(\calM_1\backslash p;\bgamma_1),\Ztilde(\calM_1/p;\bgamma_1)\big)  \begin{pmatrix}2&-1\\-1&1\end{pmatrix}
\begin{pmatrix} 1\\  s\end{pmatrix}.\end{align*}
\end{proof}

\subsection{Simple Delta-sums}
\label{sec:simplesums}

We complete this section by investigating the (simple) way in which 
the special matroids $\Itwo$ and $\Ithree$ interact with delta-sums 
with $|T|=1$ and $|T|=3$, respectively

Suppose $(\calM,\bgamma)$ is a weighted binary matroid on a ground set $E(\calM)$
of $m$~elements, and with distinguished element~$p$, which is not a loop.
Consider the delta-sum $\calM\deltasum\Itwo$.  The ground set of this 
matroid also has $m$~elements, and it shares all but one element with $E(\calM)$.
Thus, we have a natural 
correspondence between the ground sets of the two matroids.
We claim that under this correspondence the two matroids are the same,
and for this it is enough to verify that they have the same circuit space.
Note that $\{p,e\}$ is the unique non-empty cycle in~$\Itwo$.  
For any cycle $C$ in $\calM$,
exactly one of $C'=C$ or $C'=C\oplus \{p,e\}$ is 
a cycle in $\calM\deltasum\Itwo$.  The mapping $C\mapsto C'$ is invertible, and 
is the required bijection between cycles in $\calM$ and those in
$\calM\deltasum\Itwo$.  Now suppose $\delta_e=d$, and 
and let $\bgamma'$ be derived from $\bgamma$ by assigning $\gamma_p=d$.
Then $\Ztilde(\calM\deltasum\Itwo,\bgamma\deltasum\bdelta)=\Ztilde(\calM,\bgamma')$.

A similar observation applies to $\Ithree$ under delta-sum with $|T|=3$.
Suppose $(\calM,\bgamma)$ is a weighted binary matroid with distinguished elements
$T=\{p_1,p_2,p_3\}$, where $T$ is a cycle in~$\calM$.
Consider the delta-sum $\calM\deltasum\Ithree$.  
As before, there is a natural 
correspondence between the ground sets of the two matroids.
Let $C_1,C_2,C_3$ be the three 2-circuits in $\Ithree$ including 
elements $p_1,p_2,p_3$, respectively.     
Any cycle $C$ in $\calM$ can be transformed in a unique way to a cycle
$C'$ in $\calM\deltasum\Ithree$, by adding 
a subset of circuits from $\{C_1,C_2,C_3\}$.  
The mapping $C\mapsto C'$ is invertible, and is a bijection between cycles 
in $\calM$ and those in $\calM\deltasum\Ithree$.
Now let $\bgamma'$ be derived from $\bgamma$ by assigning 
$\gamma_{p_1}=\delta_{e_1}$, $\gamma_{p_2}=\delta_{e_2}$ and $\gamma_{p_3}=\delta_{e_3}$.
Then $\Ztilde(\calM\deltasum\Ithree,\bgamma\deltasum\bdelta)=\Ztilde(\calM,\bgamma')$.
 
\section{The algorithm}

We now have all the ingredients for the algorithm for estimating 
$\Ztilde(\calM;\bgamma)=\Ztilde(\calM;2,\bgamma)$, 
given a weighted regular matroid~$(\calM,\bgamma)$ and an accuracy parameter~$\epsilon$.  
The base cases for this recursive algorithm are  
when $\calM$ is graphic, cographic or $R_{10}$.  In these cases
we estimate $\Ztilde(\calM;\bgamma)$ ``directly'', which means the 
following.  If $\calM$ is $R_{10}$ then we evaluate $\Ztilde(\calM;\bgamma)$
by brute force.  If $\calM$ is graphic, we form the weighted 
graph $(G,\bgamma)$ whose (weighted) cycle matroid is $(\calM,\bgamma)$.
Then the partition function of the Ising model on $(G,\bgamma)$ may 
be estimated using the algorithm of Jerrum and Sinclair~\cite{JSIsing}.
If $\calM$ is cographic, then its dual $\calM^*$ is graphic, and
$$
\Ztilde(\calM;\bgamma)=\gamma_Eq^{-r_\calM(E)}\Ztilde(\calM^*;\bgamma^*),
$$
where $E=E(\calM)$, and $\bgamma^*$ is the dual weighting  
given by $\gamma^*_e=q/\gamma_e=2/\gamma_e$ for
all $e\in E(\calM)$~\cite[4.14a]{SokalMulti}. 
(Ground set elements~$e$ with $\gamma_e=0$ do not cause any problems, because they can just be deleted.)
Then we proceed as before, 
but using $(\calM^*,\bgamma^*)$ in place of $(\calM,\bgamma)$.
The proposed algorithm is presented as Figure~\ref{fig:alg}.
 
\begin{figure}[t]
  
\begin{framed}
\begin{description}
\item[Step 1] If $\calM$ is graphic, cographic or $R_{10}$ then estimate
$\Ztilde(\calM,\bgamma)$ directly.

\item[Step 2] Otherwise use Seymour's decomposition algorithm 
to express $\calM$ as 
$\calM_1\deltasum\calM_2$, where $\deltasum$ is a 1-, 2- or 3-sum.
Recall that $E(\calM)=E(\calM_1)\oplus E(\calM_2)$.
Let $T=E(\calM_1)\cap E(\calM_2)$, $E_1=E(\calM_1)\setminus T$ and $E_2=E(\calM_2)\setminus T$.
Noting $E(\calM)=E_1\cup E_2$, let $\bgamma_1:E_1\to\R^+$ and $\bgamma_2:E_2\to\R^+$
be the restrictions of $\bgamma$ to $E_1$ and~$E_2$. Assume without loss of generality that $|E(\calM_2)|\leq|E(\calM_1)|$
(otherwise swap their names). 

\item[Step 3] If $\deltasum$ is s 3-sum then let $T=\{p_1,p_2,p_3\}$ (say).  
Execute Steps 4a--7a in Figure~\ref{fig:alg3}.
If $\deltasum$ is a 2-sum then 
let $T=\{p\}$. Execute Steps  4b--7b in Figure~\ref{fig:alg2}.
If $\deltasum$ is a 1-sum then recursively estimate
$\Ztilde(\calM_1;\bgamma_1)$ with accuracy parameter  $\epsilon |\calM_1|/|\calM|$,
and $\Ztilde(\calM_2;\bgamma_2)$  with accuracy parameter $\epsilon |\calM_2|/|\calM|$,
and return the product of the two, which is an estimate of $\Ztilde(\calM,\bgamma)$.

\end{description} 
\end{framed}
 
\caption{Algorithm for estimating the Ising partition function of a regular matroid $\calM$ given
accuracy parameter $\epsilon\leq 1$.
}
\label{fig:alg}
\end{figure}

\begin{figure}[t]
 
\begin{framed}
\begin{description}

\item[Step 4a]  Recusively estimate $z_0=\Zdw{\calM_2}{\bgamma_2}$, $z_1=\Zow{\calM_2}{\bgamma_2}$,
$z_2=\Ztw{\calM_2}{\bgamma_2}$ and $z_3=\Zthw{\calM_2}{\bgamma_2}$
with accuracy parameter
$\epsilon \smallconst |\calM_2|/(4|\calM|)$.

\item[Step 5a] Using Lemma~\ref{lem:approx3sum},
compute $d_1$, $d_2$ and $d_3$ 
such that, for any weight 
function $\bdelta$ with $\delta_{e_1}=d_1$, 
$\delta_{e_2} = d_2$ and $\delta_{e_3} = d_3$,
 $$
e^{-
\epsilon   |\calM_2|/(2|\calM|)
 }
\Ztilde(\calM ;\bgamma )
\leq 
\easycompute \Ztilde(\calM_1 \deltasum \Ithree;\bgamma_1 \deltasum \bdelta) 
\leq 
e^{\epsilon   |\calM_2|/(2|\calM|) }
\Ztilde(\calM ;\bgamma ).$$
Note that our estimate for~$z_0$ gives an estimate  for $\easycompute= z_0/ \sqrt{R/S_1S_2S_3}$ with accuracy parameter
at most $\epsilon |\calM_2|/(2 |\calM|)$. ($R$, $S_1$, $S_2$ and $S_3$ are byproducts of the computation of $d_1$, $d_2$ and $d_3$.)

\item[Step 6a]  Recall from Section~\ref{sec:simplesums}
that  
$\Ztilde(\calM_1\deltasum\Ithree,\bgamma_1 \deltasum \bdelta)  = \Ztilde(\calM_1, \bgamma')$, 
where $\bgamma'$ is derived from $\bgamma_1$ by assigning 
$\gamma'_{p_1}=\delta_{e_1}$, $\gamma'_{p_2}=\delta_{e_2}$ and $\gamma'_{p_3}=\delta_{e_3}$.  

\item[Step 7a] Recursively estimate
$\Ztilde(\calM_1;\bgamma')$ with accuracy parameter
$\epsilon(|\calM| - |\calM_2|)/|\calM|$ and multiply it by the estimate for $\easycompute$ from Step 5a. Return this value, which is an estimate of $\Ztilde(\calM,\bgamma)$.

\end{description}
\end{framed}

\caption{The 3-sum case. 
$\smallconst$ is a sufficiently small positive constant which does not depend upon $\calM$ or $\epsilon$. See Lemmas~\ref{lem:approx3sum}
and~\ref{lem:approx2sum}.
}
\label{fig:alg3}
\end{figure}

Using the guarantees from Lemmas~\ref{lem:approx3sum} and~\ref{lem:approx2sum}, it is easy to see that the algorithm is correct.
That is, given a regular matroid~$\calM$ and an accuracy parameter $\epsilon<1$,
the algorithm returns an estimate $\hat Z$
satisfying $e^{-\epsilon}\Ztilde(\calM,\bgamma) \leq \hat Z \leq e^{\epsilon}\Ztilde(\calM,\bgamma)$.
The rest of this section shows that the running time is at most a polynomial in $|E(\calM)|$ and $\epsilon^{-1}$.

Let $\tTr(m)= O( m^{\aTr})$ be the time complexity of performing the Seymour 
decomposition of an $m$-element matroid (this is our initial splitting step), 
and let $\tJS(m,\epsilon)=O(m^\aJS \epsilon^{-2})$ 
be the time complexity of
estimating the Ising partition function of an $m$-edge graph. 
(From \cite[Theorem 5]{JSIsing}, and the remark following it, we may take
$\aJS=15$.)
Denote by $T(m,\epsilon)$ the time-complexity of the algorithm of Figure~\ref{fig:alg}.
Recall that 
$\smallconst$ is a sufficiently small positive constant which does not depend upon $\calM$ or $\epsilon$ and
is described in  Lemmas~\ref{lem:approx3sum}
and~\ref{lem:approx2sum};  we can take it to be $\smallconst=1/6000$.
The recurrence governing $T(m,\epsilon)$ is now presented, immediately
followed by an explanation of its various components.
\begin{align*}
T(m,\epsilon)&\leq \tTr(m)+\max\Big\{\tJS(m,\epsilon),\\
&\qquad\max_{4\leq k\leq m/2}\big(T(m-k+3,\tfrac{\epsilon(m-k-3)}{m})+4T(k,\tfrac{\epsilon\smallconst(k+3)}{4 m})\big),\\
&\qquad\max_{2\leq k\leq m/2}\big(T(m-k+1,\tfrac{\epsilon(m-k-1)}{m})+2T(k,\tfrac{\epsilon\smallconst(k+1)}{2 m})\big),\\
&\qquad\max_{1\leq k\leq m/2}\big(T(m-k, \tfrac{\epsilon(m-k)}{m})+T(k, \tfrac{\epsilon k}{m})\big)\Big\}.
\end{align*}
The four expressions within the outer maximisation correspond
to the direct case, the 3-sum case, the 2-sum case, and the 1-sum case,  respectively.
The variable $k$ is to be interpreted as the number of ground set elements 
in $\calM$ that come from~$\calM_2$ (and hence $m-k$ is the number that 
come from $\calM_1$).  Thus, in the case of a 3-sum, for example, $|E(\calM_1)|=m-k+3$
and $|E(\calM_2)|=k+3$.  Note that Step~2 of the algorithm ensures $k\leq m/2$.
The lower bounds on~$k$ come from the corresponding lower bounds 
on the size of matroids occurring in 3-sums, 2-sums and 1-sums.

We will demonstrate that $T(m,\epsilon)= O(m^\alpha \epsilon^{-2})$, 
where $\alpha=\max\{\aJS,\aTr+1, 43\}$.
Specifically, we will show, by induction on~$m$,
that $T(m,\epsilon)\leq Cm^\alpha\epsilon^{-2}$, for some 
constant~$C$ and all sufficiently large~$m$.  
In the analysis that follows we do not attempt to obtain the
best possible exponent~$\alpha$ for the running time, instead preferring
to simplify the analysis as much as possible.  It would certainly
be possible to reduce the constant~43 appearing in the formula for the 
exponent, but there seems little point in doing so, as the 
best existing value for $\aJS$ is already too large to make the algorithm 
feasible in practice.

\begin{figure}[t]

\begin{framed}
\begin{description}
\item[Step 4b] Recursively
estimate $z_0 =  \Ztilde(\calM_2\backslash p;\bgamma_2)$ and
$z_1 = \Ztilde(\calM_2/p;\bgamma_2)$ with accuracy parameter $(\epsilon \smallconst |\calM_2|)/(2|\calM|)$.
\item[Step 5b] Using Lemma~\ref{lem:approx2sum}, compute~$d$ such that 
for 
every weight function $\bdelta$ with $\delta_e = d$,
 $$e^{-\epsilon   |\calM_2|/(2|\calM|)}
\Ztilde(\calM ;\bgamma)
\leq\easycompute \Ztilde(\calM_1\deltasum\Itwo;\bgamma_1\deltasum\bdelta)
\leq e^{\epsilon   |\calM_2|/(2|\calM|) }
\Ztilde(\calM;\bgamma). $$
Note that our estimate for $z_0$ gives an estimate for $\easycompute = 2(2+d)^{-1} z_0$
with accuracy parameter
at most $(\epsilon |\calM_2|)/(2 |\calM|)$.

\item[Step 6b]  Recall from Section~\ref{sec:simplesums}
that  $\Ztilde(\calM_1\deltasum\Itwo,\bgamma_1\deltasum\bdelta)=\Ztilde(\calM_1,\bgamma')$, 
where $\bgamma'$ is derived from $\bgamma_1$ by assigning $\gamma'_p=d$.

 \item [Step 7b] Recursively estimate $\Ztilde(\calM_1;\bgamma')$
with accuracy parameter
$\epsilon(|\calM| - |\calM_2|)/|\calM|$ and multiply it by the estimate for $\easycompute$ from Step 5b. Return this value, which is an estimate of $\Ztilde(\calM,\bgamma)$.

\end{description}
\end{framed}

\caption{The 2-sum case.}
\label{fig:alg2}
\end{figure}

Note that by choosing $C$ sufficiently large in the time-bound 
$T(m,\epsilon)\leq Cm^\alpha\epsilon^{-2}$, we may ensure that
this bound on time-complexity bound holds for any $m$ from 
a finite initial segment of the positive integers, 
specifically for $m\in\{1,2,\ldots,19\}$.
For the inductive step,
substitute $T(m,\epsilon)\leq Cm^\alpha\epsilon^{-2}$ into the \rhs{} of 
the above recurrence.
We need to verify 
\begin{align}
Cm^\alpha\epsilon^{-2} &\geq \tTr(m)+\max\Big\{\tJS(m,\epsilon),\label{eq:easy} \\
&\qquad\max_{4\leq k\leq m/2}\big(C(m-k+3)^\alpha {\big(\tfrac{m}{\epsilon(m-k-3)}\big)}^{2} 
  +4Ck^\alpha {\big(\tfrac{4m}{\epsilon\smallconst(k+3)}\big)}^{2} \big),\label{eq:rec3sum}\\
&\qquad\max_{2\leq k\leq m/2}\big(C(m-k+1)^\alpha {\big(\tfrac{m}{\epsilon(m-k-1)}\big)}^{2}
  +2Ck^\alpha {\big(\tfrac{2 m}{\epsilon\smallconst(k+1)}\big)}^{2}\big)\label{eq:rec2sum}\\
&\qquad\max_{1\leq k\leq m/2}\big(C(m-k)^\alpha {\big(\tfrac{m}{\epsilon(m-k)}\big)}^{2}
  +Ck^\alpha {\big(\tfrac{m}{\epsilon k}\big)}^{2}\big)\label{eq:rec1sum}
\Big\},
\end{align}
for $m\geq20$,
as this will imply $T(m,\epsilon)\leq Cm^\alpha \epsilon^{-2}$,
completing the induction step.  Effectively, there are four independent 
inequalities to verify, numbered (\ref{eq:easy})--(\ref{eq:rec1sum}).
The first inequality is immediate (for a sufficiently large constant~$C$),
since we are assuming $\alpha\geq\max\{\aTr,\aJS\}$.

The second inequality, namely (\ref{eq:rec3sum}), requires some work, but the 
remaining two will then follow easily.  First, note the following 
simple estimate:
\begin{equation}\label{eq:vsimple}
\frac{m-k+3}{m-k-3}\leq\frac{m/2+3}{m/2-3}=1+\frac{12}{m-6}\leq1+\frac{20}m,
\end{equation}
where we have assumed that $k\leq m/2$ and $m\geq20$.
The inequality we wish to establish is equivalent, under rearrangement, to
$$
Cm^\alpha\epsilon^{-2}-C(m-k+3)^\alpha {\big(\tfrac{m}{\epsilon(m-k-3)}\big)}^{2} 
  -4Ck^\alpha {\big(\tfrac{4m}{\epsilon\smallconst(k+3)}\big)}^{2} \geq \tTr(m),
$$
for all $k$ with $4\leq k\leq m/2$.  Noting (\ref{eq:vsimple}), it is enough to 
show
\begin{equation}\label{eq:enough}
Cm^2\epsilon^{-2}
   \big[m^{\alpha-2}-(m-k+3)^{\alpha-2}(1+20/m)^2-64k^{\alpha-2}\smallconst^{-2}\big]
   \geq \tTr(m),
\end{equation}
for all $k$ with $4\leq k\leq m/2$.  Regarding the \lhs{} of~(\ref{eq:enough})
as a continuous function of a real variable~$k$, and taking the second derivative
with respect to~$k$, we that the \lhs{} is a concave function of~$k$.
It is enough, then, to check that (\ref{eq:enough}) holds at $k=4$ and $k=m/2$.

When $k=4$, the inequality follows from the following sequence of 
inequalities:
\begin{align}
&Cm^2\epsilon^{-2}\label{eq:k=4}
   \big[m^{\alpha-2}-(m-1)^{\alpha-2}(1+20/m)^2-64\times4^{\alpha-2}\smallconst^{-2}\big]\\
&\qquad\geq Cm^2\epsilon^{-2}\notag
   \big[m^{\alpha-2}-(m-1)^{\alpha-2}(1+1/m)^{40}-4^{\alpha+1}\smallconst^{-2}\big]\\
&\qquad=Cm^2\epsilon^{-2}\notag
   \big[m^{\alpha-2}-m^{\alpha-2}(1-1/m)^{\alpha-2}(1+1/m)^{40}-4^{\alpha+1}\smallconst^{-2}\big]\\
&\qquad\geq Cm^2\epsilon^{-2}\label{eq:alphageq43}
   \big[m^{\alpha-2}-m^{\alpha-2}(1-1/m)-4^{\alpha+1}\smallconst^{-2}\big]\\
&\qquad= Cm^2\epsilon^{-2}\notag
   \big[m^{\alpha-3}-4^{\alpha+1}\smallconst^{-2}\big]\\
&\qquad\geq \tfrac12 Cm^{\alpha-1}\epsilon^{-2}\label{eq:calc}\\
&\qquad\geq \tTr(m),\label{eq:laststep}
\end{align}
where inequality (\ref{eq:alphageq43}) is a consequence of $\alpha\geq43$, 
inequality (\ref{eq:laststep}) of $\alpha\geq\aTr+1$, and  
inequality~(\ref{eq:calc})
comes from a comparison of the two terms, noting $\alpha\geq43$, $\smallconst\geq1/6000$
and $m\geq20$.

When $k=m/2$, inequality (\ref{eq:enough}) is established as follows:
\begin{align}
&Cm^2\epsilon^{-2}\label{eq:k=m/2}
   \big[m^{\alpha-2}-(m/2+3)^{\alpha-2}(1+20/m)^2-(m/2)^{\alpha-2}(8/\smallconst)^2\big]\\
&\qquad= Cm^\alpha\epsilon^{-2}
   \big[1-(\tfrac12+3/m)^{\alpha-2}(1+20/m)^2-(1/2)^{\alpha-2}(8/\smallconst)^2\big]\notag\\
&\qquad\geq Cm^\alpha\epsilon^{-2}
   \big[1-10^{-7}-10^{-2}]\label{eq:calc'}\\
&\qquad\geq\tfrac12 Cm^\alpha\epsilon^{-2}\notag\\
&\qquad\geq \tTr(m),\label{eq:laststep'}
\end{align}
where (\ref{eq:laststep'}) is a consequence of $\alpha\geq\aTr$,
and~(\ref{eq:calc'}) of 
$\alpha\geq43$, $\smallconst\geq1/6000$
and $m\geq20$.

The above calculations may easily be adapted to cover inequalities
(\ref{eq:rec2sum}) and~(\ref{eq:rec1sum}) for the cases of 2-sum and 1-sum.
Since estimate (\ref{eq:vsimple}) applies as well to $(m-k+1)/(m-k-1)$, 
the analogue of (\ref{eq:enough}) in the 2-sum case is
$$
Cm^2\epsilon^{-2}
   \big[m^{\alpha-2}-(m-k+1)^{\alpha-2}(1+20/m)^2-8k^{\alpha-2}\smallconst^{-2}\big]
   \geq \tTr(m).
$$
Again, we need only verify this at the extreme values of~$k$, namely $k=2$ and
$k=m/2$.
The specialisation to $k=2$,
$$
Cm^2\epsilon^{-2}
   \big[m^{\alpha-2}-(m-1)^{\alpha-2}(1+20/m)^2-2^{\alpha+1}\smallconst^{-2}\big]
   \geq \tTr(m),
$$
can be seen by comparison with (\ref{eq:k=4}), and the one for $k=m/2$,
$$
Cm^2\epsilon^{-2}
   \big[m^{\alpha-2}-(m/2+1)^{\alpha-2}(1+20/m)^2-8(m/2)^{\alpha-2}\smallconst^{-2}\big]
   \geq \tTr(m).
$$
by comparison with (\ref{eq:k=m/2}).

Finally, the analogue of (\ref{eq:enough}) in the 1-sum case is
$$
Cm^2\epsilon^{-2}
   \big[m^{\alpha-2}-(m-k)^{\alpha-2}-k^{\alpha-2}\big]
   \geq \tTr(m),
$$
which again needs to be verified at the extreme values $k=1$ and $k=m/2$.
As before, these can be seen by comparison with (\ref{eq:k=4}) and~(\ref{eq:k=m/2}).

\bibliographystyle{plain}
\bibliography{matroidising}

\section{Appendix}

This appendix contains some technical lemmas needed in the proof of Lemma~\ref{lem:approx3sum}.
The lemmas are about approximation, and they don't add any intuition to the paper.

We start with a well-known fact that we will use in the proof of both lemmas. (This follows directly from the series expansion of $e$.)
\begin{observation}
\label{fact}
If $0<\epsilon<1$ then $1+\epsilon \leq e^\epsilon \leq 1+2\epsilon$. 
\end{observation}

Our first lemma refers to the matrix~$D$ defined just before Lemma~\ref{lem:3sumsplit}.

\begin{lemma}
\label{lem:YYY}
Suppose that $z$, $s$ and $r$ are column vectors in $\mathbb{R}^4$ with positive entries satisfying
$1 \leq z_i / z_0 \leq 2$,
$1 \leq s_i / s_0 \leq 2$, and $1 \leq r_i/r_0 \leq 2$ for $i\in\{1,2,3\}$. 
Suppose that $e^{-\epsilon} s_i \leq r_i \leq e^{\epsilon} s_i$ for some $0<\epsilon<1$.
Then $e^{- 56 \epsilon} z^\tr D s \leq z^\tr D r \leq e^{ 56 \epsilon} z^\tr D s$. 
\end{lemma}

\begin{proof}
 
For any positive column vector~$v$ in~$\mathbb{R}^4$ satisying 
$1 \leq v_i/v_0$,
$$ z^\tr D v = z_0 v_0 
+(z_1-z_0)(v_1-v_0)
+(z_2-z_0)(v_2-v_0)
+(z_3-z_0)(v_3-v_0) \geq z_0 v_0.$$
Also, summing the absolute values of all of the monomials in $z^\tr D v$,
and using  $z_i/z_0\leq 2$ and $v_i/v_0\leq 2$,
we get $\sum_{i,j} |D_{i,j}| z_i v_j   \leq 28 z_0 v_0   $.
Then 

\begin{align*} z^\tr D s  - z^\tr D r  &= \sum_{i,j} D_{i,j} z_i (s_j-r_j) \\
                                             &\leq \sum_{i,j} |D_{i,j}| z_i (e^{\epsilon}-1) r_j \\
                                             &\leq 2 \epsilon \sum_{i,j} |D_{i,j}| z_i r_j\\
                                             & \leq 56 \epsilon z_0 r_0 \\
                                             & \leq 56 \epsilon z^\tr D r,\end{align*}
                                             
so $z^\tr D s \leq (1+56 \epsilon) z^\tr D r \leq e^{56 \epsilon} z^\tr D r$.                                           
The inequality
$z^\tr D r \leq   e^{56 \epsilon} z^\tr D s$ follows by interchanging the roles of~$r$ and~$s$ in the proof.                                           
\end{proof}

Our second lemma refers to the equations in Lemma~\ref{lem:simsig3}.

\begin{lemma}
\label{lem:XXX}
Suppose that $r_1 \leq r_2 \leq r_3$ satisfy the following equations.
\begin{align}
\label{E1} 2+ r_1 - r_2 - r_3 & > 0,\\
\label{E2} r_1 + r_2 + r_3 - r_2 r_3 -2 &\geq 0,\\
\label{E3} 1 \leq r_i \leq 2, \mbox{ for $i\in\{1,2,3\}$}.
\end{align}
Suppose that we are given values $\tilde s_1 \leq \tilde s_2 \leq \tilde s_3$
satisfying 
$ e^{-\chi  } r_i \leq \tilde s_i 
\leq e^{\chi  } r_i$ for $i\in\{1,2,3\}$, where $\chi$ is a sufficiently small positive constant.
Using $\tilde s_1$, $\tilde s_2$ and $\tilde s_3$,
we can compute values $s_1$, $s_2$ and $s_3$ satisfying $1\leq s_i\leq 2$,
$$  e^{-\newbigconst \chi  } s_i \leq  r_i \leq e^{\newbigconst \chi } s_i, $$
and Equations~(\ref{eq:sig1}), (\ref{eq:sig2}) and (\ref{eq:sig3}).

\end{lemma}

\begin{proof}
Let $\delta = 4 e \chi  $. Note that 
\begin{equation}
\label{startclose}
| \tilde s_i - r_i| \leq \delta.
\end{equation}
For example,
using Observation~\ref{fact},  
$r_i - \tilde s_i \leq 2 \chi  \tilde s_i \leq 2 \chi   e^{\chi  } r_i$.
Since $r_i \leq 2$  (by Equation (\ref{E3})) and $\chi \leq 1$,
$r_i - \tilde s_i \leq \delta$.  Similarly, 
$\tilde s_i - r_i \leq \delta$.

In each of  two cases, we will compute  
$s_1$, $s_2$, and $s_3$ so that

\begin{align}
\label{EE0} 1\leq s_1 \leq s_2 \leq s_3 \leq 2,\\
\label{EE1} 2+ s_1 - s_2 - s_3 & > 0,\\
\label{EE2} s_1 + s_2 + s_3 - s_2 s_3 -2 &\geq 0,\\
\label{finalclose} |s_i - r_i| \leq 6 \delta.\end{align}
Equations~(\ref{EE0}), (\ref{EE1}) and (\ref{EE2}) imply 
Equations~(\ref{eq:sig1}), (\ref{eq:sig2}) and (\ref{eq:sig3}).
Also, Equation~(\ref{finalclose})  
  implies $ e^{-\newbigconst \chi  } s_i \leq  r_i \leq e^{\newbigconst \chi } s_i$
since
$ s_i \leq r_i + 6 \delta  \leq r_i(1+6 \delta) \leq r_i e^{6\delta} \leq r_i e^{\newbigconst \chi}$
and similarly
$ r_i \leq   s_i e^{\newbigconst \chi}$.
 
 {\bf Case 1: }  $\tilde s_2 - \tilde s_1 \leq 5 \delta$.

Take $s_1 = s_2 = \min(\max(1,\tilde s_2),2-\delta)$ and $s_3 = \min(\max(1,\tilde s_3),2-\delta)$.
(\ref{EE0}) follows easily from the definitions since $2-\delta\geq 1$ and $\tilde s_2 \leq \tilde s_3$.
(\ref{EE1}) and (\ref{EE2}) follow from the facts that $s_1=s_2$, $s_1\geq 1$ and $s_3<2$ since, 
$s_1 + s_1 + s_3 - s_1 s_3 -2 = (s_1-1)(2-s_3)$.
To establish (\ref{finalclose}) 
note that
$$
s_1 \leq\max(1,\tilde s_2) \leq \max(1,\tilde s_1 + 5 \delta) \leq \max(1,r_1+6 \delta) \leq r_1 + 6\delta.
$$
Also, since $\tilde s_2 \leq r_2+\delta \leq 2+\delta$,
$$ s_1 \geq \tilde s_2 - 2\delta \geq \tilde s_1 - 2 \delta \geq r_1 - 3\delta.$$
Similarly, $r_2 - 3 \delta \leq s_2 \leq r_2 + \delta$ and $r_3 - 3 \delta \leq s_3 \leq r_3+\delta$.

{\bf Case 2: } $\tilde s_2 - \tilde s_1 > 5 \delta$.

 Take $s_1 = \tilde s_1 + 4 \delta$, 
$s_2 = \min(\tilde s_2,2)$ and 
$s_3 = \min(\tilde s_3,2)$.  
Equation (\ref{EE0}) follows  since $\tilde s_1 + 4 \delta \leq \tilde s_2 - \delta \leq s_2$.
(\ref{EE1}) follows
since
\begin{align*}
2+ s_1 - s_2 - s_3 &\geq 2 + \tilde s_1 + 4 \delta - \tilde s_2 - \tilde s_3\\
&\geq 2 +  (r_1-\delta) + 4 \delta -  (r_2+\delta) -  (r_3+\delta)\\
&\geq   \delta + 2+ r_1 - r_2 - r_3\\
&>  \delta.
\end{align*}

We now show that (\ref{EE2}) holds.
Note that $$s_1 + s_2 + s_3 - s_2 s_3 -2 = s_1 - (s_2-1)(s_3-1)-1,$$ and
the this quantity is increasing as a function of~$s_1$ and
decreasing as a function of~$s_2$ and as a function of~$s_3$.
Also, 
\begin{align*} s_1 &= \tilde s_1 + 4 \delta \geq r_1 + 3 \delta,\\
s_2 &= \min(\tilde s_2,2) \leq \tilde s_2 \leq r_2 + \delta, \mbox{ and }\\
s_3 &\leq r_3 + \delta.
\end{align*}
So 
\begin{align*} 
s_1 + s_2 + s_3 - s_2 s_3 -2 
&\geq (r_1 + 3 \delta) - (r_2+ \delta-1)(r_3 + \delta-1) -1
\\
&= 3 \delta - \delta(r_2-1)-\delta(r_3-1) - \delta^2 +
\left(r_1 - (r_2-1)(r_3-1)-1\right)
\\
 &\geq 3 \delta - \delta(r_2-1)-\delta(r_3-1) - \delta^2, 
 \end{align*}
 and this is at least~$0$, since $r_2-1\leq 1$, $r_3-1\leq 1$, and $\delta\leq 1$.
 
 To establish (\ref{finalclose}), note that
 $s_1 = \tilde s_1 + 4 \delta
 \leq r_1 + 5 \delta$ 
 and $r_1 \leq s_1$.
 Also,  as noted above, $\tilde s_2 - \delta \leq s_2 \leq \tilde s_2$
 and similarly $\tilde s_3 - \delta \leq s_3 \leq \tilde s_3$.
 
   \end{proof}

\end{document}